 \documentclass[final,5p,times,twocolumn]{elsarticle}
\usepackage{graphicx}
\usepackage{subcaption}
\usepackage{multirow}
\usepackage{wrapfig}
\usepackage{array}

\usepackage{svg}

\usepackage{tabularray}
\usepackage{tabularx}
\usepackage{rotating}
\usepackage{makecell}
\usepackage{booktabs}

\usepackage[utf8]{inputenc}
\usepackage{amsmath}
\usepackage{amsthm}
\usepackage{amsfonts}
\usepackage{epsfig}
\usepackage{psfrag}
\usepackage{pstricks}
\usepackage{algorithm}
\usepackage{stfloats}

\usepackage{url}

\usepackage{cuted}

\newtheorem{theorem}{Theorem}

\newtheorem{proposition}[theorem]{Proposition}

\usepackage{algorithmicx}
\usepackage{algpseudocode}  
\floatname{algorithm}{Algorithm}

\graphicspath{{./Figures/}}
\usepackage{amssymb}
\usepackage{bm}
\usepackage{hyperref}
\hypersetup{%hypertex=true,
            colorlinks=true,
            linkcolor=blue,
            anchorcolor=blue,
            citecolor=blue
            }
\biboptions{comma,sort&compress}

\usepackage{enumitem}
\setlist{nosep,topsep=-\parskip}

\usepackage{ulem}

\newcommand{\rev}[2]{#2} %the revised version

\journal{SPM}

\begin{document}

\captionsetup[figure]{labelfont={},name={Figure}}

\begin{frontmatter}

\title{Translatiing TPMS models to STEP files}
\author[cad]{Yaonaiming Zhao}
\author[cad]{Qiang Zou\corref{cor}}\ead{qiangzou@cad.zju.edu.cn}

\cortext[cor]{Corresponding author.}
\address[cad]{State Key Laboratory of CAD$\&$CG, Zhejiang University, Hangzhou, 310027, China}
\address[Aviation]{Aviation Industry Shenyang Aircraft Design and Research Institute, Liaoning, 110035, China}

\begin{abstract}
Triply periodic minimal surface (TPMS) is emerging as an important way of designing microstructures. However, there has been limited use of commercial CAD/CAM/CAE software packages for TPMS design and manufacturing. This is mainly because TPMS is consistently described in the functional representation (F-rep) format, while modern CAD/CAM/CAE tools are built upon the boundary representation (B-rep) format. One possible solution to this gap is translating TPMS to STEP, which is the standard data exchange format of CAD/CAM/CAE. Following this direction, this paper proposes a new translation method with error-controlling and $C^2$ continuity-preserving features. It is based on an approximation error-driven TPMS sampling algorithm and a constrained-PIA algorithm. The sampling algorithm controls the deviation between the original and translated models. With it, an error bound of $2\epsilon$ on the deviation can be ensured if two conditions called $\epsilon$-density and $\epsilon$-approximation are satisfied. The constrained-PIA algorithm enforces $C^2$ continuity constraints during TPMS approximation, and meanwhile attaining high efficiency. A theoretical convergence proof of this algorithm is also given. The effectiveness of the translation method has been demonstrated by a series of examples and comparisons.
\end{abstract}


\end{frontmatter}

%\linenumbers

\section{Introduction}
\label{sec:intro}
The triply periodic minimal surface (TPMS) is a porous cellular-like surface with zero mean curvature~\cite{2019_fang_tpms_definition}. TPMS structures have advantages such as high porosity, high surface-to-volume ratios, and super mechanical properties~\cite{2022_wang_tpms_property_porous}, etc. With these features, as well as recent advancements in additive manufacturing technologies~\cite{2021_ding_tpms_slicing}, come increasing applications of TPMS structures in fields like aerospace~\cite{2023_careri_tpms_application_aerospace}, biomedical~\cite{2014_yang_tissue_engineering}, and energy~\cite{2019_catchpole_tpms_application_energy_thermal_conductivity}.

The design and manufacturing of TPMS structures require tight integration of computer-aided design, manufacturing, and engineering (CAD/CAM/CAE) tools~\cite{liu2021memory}. However, TPMS is usually described in the functional representation (F-rep) format~\cite{2023_hong_tpms_functional_representation}, while modern CAD/CAM/CAE software packages are built upon the boundary representation (B-rep) format~\cite{zou2023variational}. There are two possible solutions to this problem: (1) F-rep-based integration using specialized CAD/CAM/CAE software~\cite{2022_wang_tpms_property_porous,2022_distance-field-rep_TPMS,2023_neural-implicit-rep_TPMS,2023_hong_tpms_functional_representation}; (2) STEP-based integration using existing commercial CAD/CAM/CAE software~\cite{2011_TPMS-to-mesh_visualization,2021_TPMS-to-mesh_analysis,2021_Flores_TPMS_NURBS_generation_Gyroid}. (STEP stands for the Standard for The Exchange of Product data, and it is a B-rep-based neutral file format supporting the solid model exchanges among CAD/CAM/CAE tools~\cite{zou2019push,zou2022robust,wang2023quasi,zou2014iso}.)

Developing brand-new, specialized CAD/CAM/CAE software that can run stably and efficiently in real industrial practice requires huge investment and a long time of fine-tuning and user training~\cite{2016_xiao_tspline_data_exchange}. This work thus opts for the second strategy, which allows TPMS structures to be directly designed and made in existing CAD/CAM/CAE software. Then, the conventional design and manufacturing tools like geometric modeling~\cite{zou2020decision,zou2019variational,su2020accurate}, design optimization~\cite{li2023xvoxel,martins2021engineering}, and tool path generation~\cite{zou2013iso,zou2021length,zou2021robust,luo2023simple,wang2023computing} become applicable to TPMS. Nevertheless, this cannot happen without a tool that is able to translate TPMS models to STEP files. Existing methods related to this topic may be classified into two major categories. The first one is to approximate TPMS with surface meshes~\cite{2021_TPMS-to-mesh_analysis,2011_TPMS-to-mesh_visualization}. 
The second category converts TPMS to Non-Uniform Rational B-Splines (NURBS) surfaces~\cite{2021_Flores_TPMS_NURBS_generation_Gyroid}. Compared to meshes, NURBS surfaces offer better accuracy and compactness in representing shapes~\cite{2022_wu_NURBS_advantages}. However, existing TPMS-to-NURBS methods are consistently limited to simple cases, without considering error control over the approximation process and continuity preservation in the translated model~\cite{2021_Flores_TPMS_NURBS_generation_Gyroid}. Without controlled errors and smoothness, the translated model cannot be used by downstream applications with high fidelity.

This paper presents our attempts to attain error control and $C^2$ continuity in TPMS-to-STEP translation for the most commonly used Gyroid, Diamond, and Schwarz\_P TPMS structures. The proposed method can guarantee that the deviation between the original structure and the translated solid model will not exceed a given error threshold. That is, an error bound of $2\epsilon$ on the deviation can be ensured if two conditions called $\epsilon$-density and $\epsilon$-approximation are satisfied. Besides error controlling, the proposed method can preserve the $C^2$ continuity of translated models. $C^2$ continuity constraints are first derived and then incorporated into the traditional progressive iterative approximation (PIA) algorithm~\cite{2018_lin_survey_pia} to attain an extended version of PIA, called constrained-PIA or simply CPIA. CPIA not only provides the $C^2$ continuity feature but also has high computational efficiency. This is particularly useful for cases requiring high translation accuracy, which implies a large number of sampling points to be approximated. We also prove that CPIA has a guaranteed convergence. Note that CPIA is general enough to handle a higher order of continuity, but this work focuses primarily on $C^2$ continuity.

The following sections begin with a review of related work in Sec.~\ref{sec:related_work}, then the proposed method's details in Sec.~\ref{sec:methods}, then the validation of the method using a series of examples and comparisons in Sec.~\ref{sec:results}, and finally conclusions on the method’s advantages and limitations in Sec.~\ref{sec:conclusion}.

\section{Related Work}
\label{sec:related_work}
This section briefly discusses the literature related to TPMS modeling (Sec.~\ref{sec:TPMS_modeling}) and NURBS approximation (Sec.~\ref{sec:NURBS_approximation}). The TPMS modeling methods are classified into two categories: implicit modeling and explicit modeling. The process of NURBS approximation involves four steps: parameterization, knot placement, weights assignment, and control points computation. Existing methods for each step are reviewed.

\subsection{TPMS modeling}
\label{sec:TPMS_modeling}
\textbf{Implicit methods.} Implicit models use mathematical functions to represent the shape of TPMS structures. Various mathematical functions have been used for this purpose, including trigonometric equation-based method\rev{~\cite{2019_fang_tpms_definition}}{~\cite{2019_fang_tpms_definition, 2021_hu_tpms_modeling_implicit}}, signed distance field-based method~\cite{2022_distance-field-rep_TPMS}, and boundary curves-based method\rev{~\cite{2023_neural-implicit-rep_TPMS}}{~\cite{2023_neural-implicit-rep_TPMS, 2023_liane_tpms_modeling_boundary_curve_based}}, among others. Despite wide applications, these methods have some serious limitations. For example, because of the limited expressiveness of trigonometric equations, the trigonometric equation-based method can only express a small part of TPMS structures. The signed distance field-based method is general, but when it is converted to triangular meshes, sharp features are often lost and low-quality meshes are generated. What's worse, implicit TPMS models may still need to be explicitized  (i.e., to meshes) to make downstream applications like interaction, analysis, and fabrication easier~\cite{2021_TPMS-to-mesh_analysis,2011_TPMS-to-mesh_visualization}. Mesh-free interaction, analysis, and fabrication of TPMS models are possible, and considerable progress has been made, e.g.,~\cite{2021_ding_tpms_slicing,2023_jiang_meshless_tpms_based_analysis}. This direction of research is promising, but they are essentially developing another collection of CAD/CAM/CAE tools that completely leave B-rep techniques out, which is never easy and requires a huge investment and a long time of fine-tuning and user training~\cite{2016_xiao_tspline_data_exchange}.

\textbf{Explicit methods.} In contrast to implicit methods, explicit methods can make full use of existing CAD/CAM/CAE software packages since they are B-rep-based. Over the past few decades, various B-rep schemes have been employed to model TPMS, including mesh surface-based method~\cite{2021_TPMS-to-mesh_analysis,2019_feng_tpms_mesh_visualization,2021_asbai_tpms_to_mesh,2022_kestutis_TPMS_mesh,2024_Na_TPMS_meshing}, subdivision surface-based method~\cite{2019_savio_subdivision_surface_TPMS,2021_mesh_subdivision_tpms,2011_Pan_minimal-subdivision-surface}, and NURBS surface-based method~\cite{2021_Flores_TPMS_NURBS_generation_Gyroid}. Among them, the NURBS-based methods offer more potential for interfacing with existing CAD/CAM/CAE tools due to the fundamental role of NURBS curves/surfaces in CAD/CAM/CAE~\cite{2015_aubry_nurbs_application,2020_noruzi_nurbs_application}. In particular, the method presented by Flores et al.~\cite{2021_Flores_TPMS_NURBS_generation_Gyroid} decomposes the problem of approximating TPMS models with NURBS surfaces into several sub-problems of approximating a portion of TPMS models and then assembling them to attain the overall approximation. This way of working is easier than the original problem. Unfortunately, only simple cases were considered, and how to control approximation error and preserve continuity in assembling individual approximations is not quite known. The proposed method also employs this decomposition-then-recombination framework but with new error-controlling and $C^2$ continuity-preserving features.

\subsection{NURBS approximation}
\label{sec:NURBS_approximation}
The underlying technique of TPMS-to-STEP translation is NURBS surface approximation, which is in turn related to parameterization, knot placement, weight assignment, and computation of control points. For this reason, this work also gives a brief review of each of these procedures, as follows. A thorough review of them can be found in existing literature, e.g.,~\cite{1995_NURBS_BOOK}.

\textbf{Parameterization.}
Parametrization means associating each data point with a parameter value. Heuristic methods are the dominant means of parameterization, e.g., uniform~\cite{1995_NURBS_BOOK}, chord length~\cite{2013_fang_chord_length}, centripetal~\cite{1989_lee_centripetal_parameterization}, universal~\cite{1999_lim_universal_parameterization}, and Foley-Nielson~\cite{1989_foley_knot}, to name a few. Such heuristic methods have also been used as an initial guess for further parameterization, e.g.,~\cite{2016_iglesias_parameterization_optimize}.\rev{However, this leads to a highly non-linear constrained optimization problem. As a result, the above heuristic methods are usually the parameterization method of choice in practice.}{}

\textbf{Knot placement.}
Knot placement refers to the choice of the number of knots and their locations for constructing B-spline bases. Like parameterization, heuristic methods like the averaging technique (AVG) and the knot placement technique (KTP) are often the method of choice in practice~\cite{1995_NURBS_BOOK,piegl2000surface}. These methods have also been used as an initial guess and then optimized through iterative knot insertion~\cite{2019_luo_knot_calulation, 2021_michel_heuristic_knot_placement}. \rev{The essential task here is to determine when and where to insert the new knots. Several effective heuristics have been developed, such as angle variations[39], and minimum knot size[40].}{}
% \rev{The essential task here is to determine when and where to insert the new knots. Several effective heuristics have been developed, such as angle variations~\cite{2005_li_adaptive_knot_placement}, and minimum knot size~\cite{2015_kang_knot_placement_sparse_optimization}.}{}

\textbf{Weight assignment.} Weights are used to attain more expressibility in representing shapes with NURBS. Weights are usually added in an ad-hoc manner, e.g., the uniform method~\cite{1995_NURBS_BOOK} and the curvature-based method~\cite{1992_Farin_Curvature_based_weight_assignment}. \rev{In the uniform method, each control point is assigned with the same weight, which is simple to implement but loses the benefits from using weights. The curvature-based method assigns weights to control points based on the local curvature. This method allows precise control to capture geometry details, but it is time-consuming.}{}

\rev{}{The parameterization, knot placement, and weight assignment methods mentioned above provide a variety of options in NURBS approximation. In this paper, to simplify the NURBS approximation problem in TPMS2STEP and make it easier to implement, uniform parameterization, uniform knot placement, and uniform weight assignment methods are chosen.}

\textbf{Control points computation.} \rev{Since control points directly determine the shape of surfaces, their calculation is crucial in NURBS approximation.}{} The control points computation is typically expressed in terms of a least square optimization problem, which boils down to solving a linear system~\cite{2021_Flores_TPMS_NURBS_generation_Gyroid}. This method is time-consuming and becomes infeasible when the number of the data points to be approximated is very large~\cite{2018_lin_survey_pia}.
To address this issue, iterative methods have been used instead, such as those presented in~\cite{2008_Nils_gauss_newton_based_iterative_solving_linear_systems,2005_Lin_pia_for_NURBS}. Among them, the progressive iterative approximation (PIA) method~\cite{2018_lin_survey_pia} may be the easiest to implement and provides more flexibility to incorporate geometric constraints. Thus, it will be used as the backbone algorithm for developing the error-controlled and $C^2$ continuity-preserving TPMS approximation algorithm in this paper.

As can be seen from the above review, NURBS approximation is a well-established field, as long as the data points are ready and of good quality. In this regard, this work does not intend to add new results to NURBS approximation but to adapt the existing methods stated above to the special needs of TPMS-to-STEP translation. That is, we focus primarily on two problems: (1) how to prepare quality data points for NURBS approximation so that the final results have controlled errors; (2) how to ensure that the continuity of original TPMS models is not broken in the approximation.

\section{Methods}
\label{sec:methods}
\subsection{TPMS and NURBS preliminaries}
\label{sec:preliminaries}
To better present TPMS2STEP, some basic knowledge of TPMS and NURBS, as well as their notations, is first described. TPMS refers to a surface with zero mean curvature arranged periodically in 3D space. To design a TPMS solid model with non-zero thickness (Fig.~\ref{tpms_thickness}), the equations for offset primitive surfaces can be derived from the Enneper–Weierstrass representation~\cite{1995_Gandy_TPMS_explicit_representation_Weierstrass_Diamond,2000_Gandy_TPMS_explicit_representation_Weierstrass_Gyroid,2000_Gandy_TPMS_explicit_representation_Weierstrass_SchwarzP} :
\begin{align}\label{offsetsurface}
\begin{aligned}
    & x=e^{i\theta }Re \left[ \int_{\omega _0}^{\omega }(1-\tau^2) R(\tau)d\tau \right]+\frac{n_xd}{\sqrt[]{n_x^2+n_y^2+n_z^2} } \\
    & y=e^{i\theta }Re\left[ \int_{\omega _0}^{\omega }i(1+\tau^2) R(\tau)d\tau \right]+\frac{n_yd}{\sqrt[]{n_x^2+n_y^2+n_z^2} }\\
    & z=e^{i\theta }Re\left[ \int_{\omega _0}^{\omega }2\tau R(\tau)d\tau \right]+\frac{n_zd}{\sqrt[]{n_x^2+n_y^2+n_z^2} }
\end{aligned}
\end{align}
where $i$ is an imaginary unit. $ Re$ extracts the real part of a complex number. $ d$ is the offset value. $ \theta$ is the Bonnet angle. $R(\tau)$ is the Weierstrass function with a complex variant $\tau$. $ \omega_0$ and $ \omega$ are the two end points of the integral path in Enneper–Weierstrass representation. The normal vectors $n=(n_x,n_y,n_z)$ are defined as\rev{}{~\cite{2004_mit_course_chapter_18}}:
\begin{align*}
\begin{aligned}
    & n_x=2Im(\phi_2\overline{\phi_3}) \\
    & n_y=2Im(\phi_3\overline{\phi_1}) \\
    & n_z=2Im(\phi_1\overline{\phi_2})
\end{aligned}
\end{align*}
where $\phi_1$, $\phi_2$, and $\phi_3$ are the first-order derivatives of Enneper–Weierstrass representation. $ \overline{\phi_1}$, $ \overline{\phi_2}$, and $ \overline{\phi_3}$ are the conjugation of $\phi_1$, $\phi_2$, and $\phi_3$ respectively. \rev{}{$Im$ refers to the operator taking the imaginary part of a complex number.}

\begin{figure}[t]
  \centering
  \includegraphics[width=0.45\textwidth]{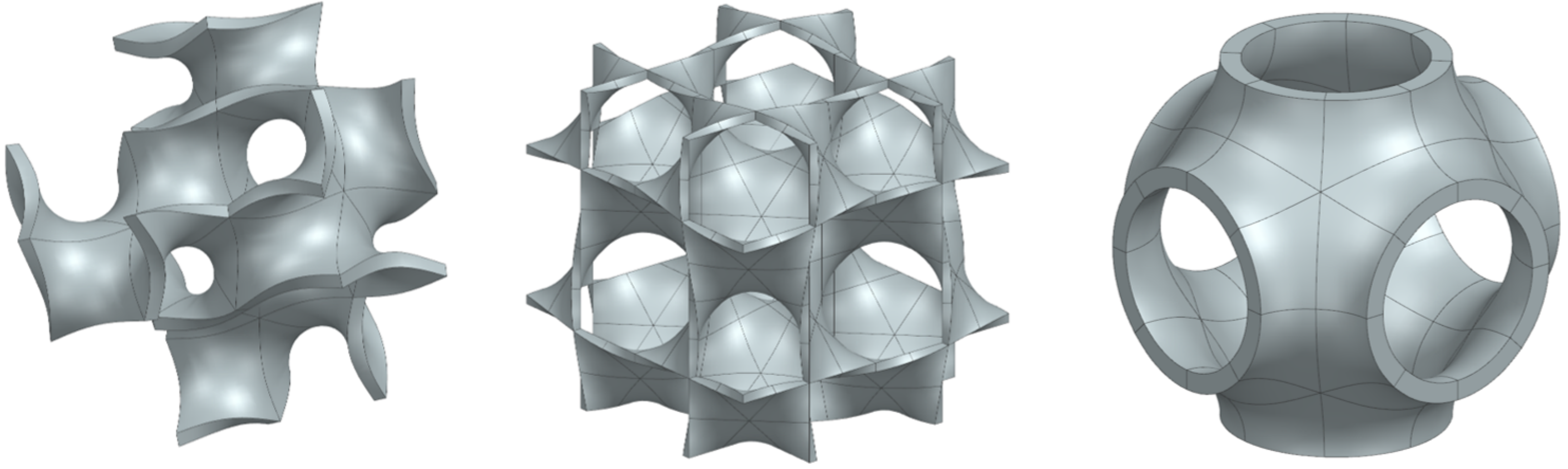}
  \caption{Illustration of TPMS solid models for Gyroid (left), Diamond (middle), and Schwarz\_P (right).}
  \label{tpms_thickness}
\end{figure}

A NURBS surface is described by:
\begin{equation}
    \mathbf{S}(u, v) = \frac{\sum_{i=0}^{n}\sum_{j=0}^{m} \mathbf{B}_{i,p}(u)\mathbf{B}_{j,q}(v)w_{i,j}\mathbf{P}_{i,j}}{\sum_{i=0}^{n}\sum_{j=0}^{m} \mathbf{B}_{i,p}(u)\mathbf{B}_{j,q}(v)w_{i,j}}
\end{equation}
where $ u$ and $ v$ are parameters in $ [0,1]$. $\mathbf{B}_{i,p}(u), i=0,\cdots,n,$ and $\mathbf{B}_{j,q}(v), j=0,\cdots,m,$ are the B-spline basis functions of degree p and q. $ \mathbf{P}_{i,j}, i=0,\cdots,n, j=0,\cdots,m$ are control points. $w_{i,j}$ means the weight of the corresponding control point $ \mathbf{P}_{i,j}$. The basis functions $\mathbf{B}_{i,p}(u)$ and $\mathbf{B}_{j,q}(v)$ are generated using a knot vector $U=\{0=u_0<u_1<\cdots<u_n=1\}$ in $u$ direction and a knot vector $V=\{0=v_0<v_1<\cdots<v_m=1\}$ in $v$ direction, as follows:
\begin{align*}
    & \mathbf{B}_{i,0}(u)=\left\{\begin{array}{rcl} 
    1 & if \;u_i\le u\le u_{i+1}\\
    0 & otherwise
    \end{array}\right . \\
    & \mathbf{B}_{i,p}(u)=\frac{u-u_i}{u_{i+p}-u_i} \mathbf{B}_{i,p-1}(u)+\frac{u_{i+p+1}-u}{u_{i+p+1}-u_{i+1}} \mathbf{B}_{i+1,p-1}(u)
\end{align*}

\subsection{The TPMS2STEP pipeline}
\label{sec:pipeline}
The proposed TPMS-to-STEP translation pipeline is shown in Fig.~\ref{tpms_pipeline}. It begins with four input parameters: the type of TPMS, the offset value for turning a TPMS surface into a solid, the approximation tolerance, and possibly the scaling factor varying with location. This input can be viewed as a parametric representation of the intended TPMS structure. This parametric representation is converted to NURBS surfaces with regard to the approximation tolerance specified by the user and the $C^2$ continuous constraints intrinsic to TPMS itself. This is where the proposed CPIA comes into play. Because TPMS is highly periodic, there is no need for approximating all surfaces bounding the TPMS structure but only several representative surfaces (to be called primitive surfaces in what follows). Having primitive surfaces in place, we create several copies of them and reposition them to the proper locations, followed by a sewing step to form a watertight, $C^2$ continuous B-rep model. The last step is to apply the varying scaling factor, if any, to this B-rep model and then write the scaled model to a .step file according to the STEP format.

\begin{figure*}[th!]
  \centering
  \includegraphics[width=1.0\textwidth]{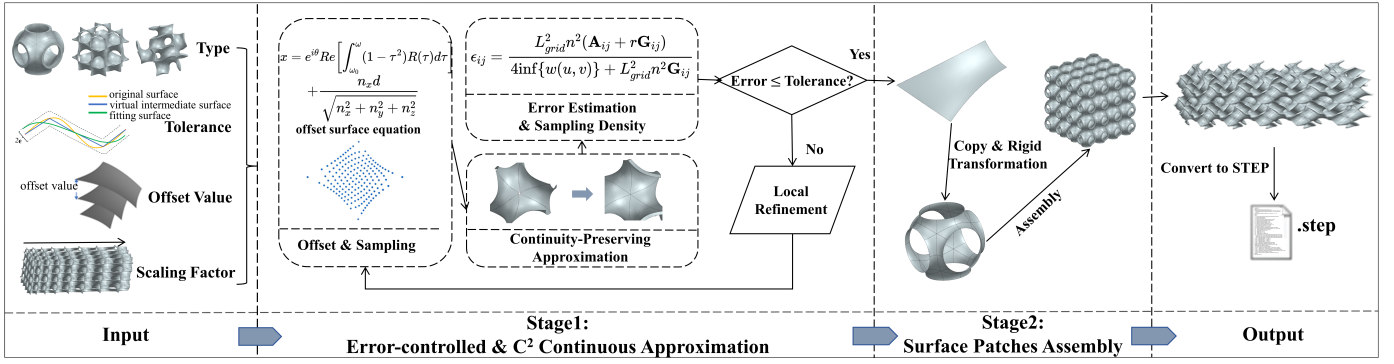}
  \caption{The proposed TPMS-to-STEP translation pipeline.}
  \label{tpms_pipeline}
\end{figure*}

In the above process of constructing NURBS surfaces, the error between the original surface and the approximation surface needs to be controlled so that the resulting STEP file can be used by downstream applications with high fidelity. \rev{An error-controlled sampling method is proposed for this purpose. There are two types of errors involved: one from sampling density, and the other from the approximation algorithm used. These two errors are actually coupled, and directly controlling the overall error from the final approximation surface to the original TPMS model is hard. To solve this challenge, we first create a virtual intermediate surface, which allows partial error control. That is, if the error from the original TPMS model to this intermediate surface is $\epsilon$, and the other error from the final approximation surface to this intermediate surface is $\epsilon$, then the overall error must not exceed $2\epsilon$, as shown in Fig.~\ref{error_source}. Then this error bound $2\epsilon$ can be forced to equal the user-specified tolerance. As such, the two errors are decoupled, and the error-controlling task is made easier. For each partial error control, an error estimation model is also developed so that the sampling can be not only driven by the user-specified tolerance but also done adaptively.}{Because it is hard, if not impossible, to simultaneously carry out TPMS sampling and NURBS fitting with respect to an error limit, a virtual intermediate surface is created to divide the overall error into two decoupled errors: one from the TPMS samples to the virtual intermediate surface, and the other from the virtual intermediate surface to the approximate NURBS surface. Such a decoupling of errors allows partial error control. That is, if the error from the original TPMS model to this intermediate surface is $\epsilon$, and the other error from the final approximation surface to this intermediate surface is $\epsilon$, then the overall error must not exceed $2\epsilon$, as shown in Fig.~\ref{error_source}. Then this error bound $2\epsilon$ can be forced to equal the user-specified tolerance, thereby simplifying the error-controlling task. For each partial error control, an error estimation model is to be developed so that the sampling can be not only driven by the user-specified tolerance but also done adaptively.}

\begin{figure}[t]
  \centering
  \includegraphics[width=0.45\textwidth]{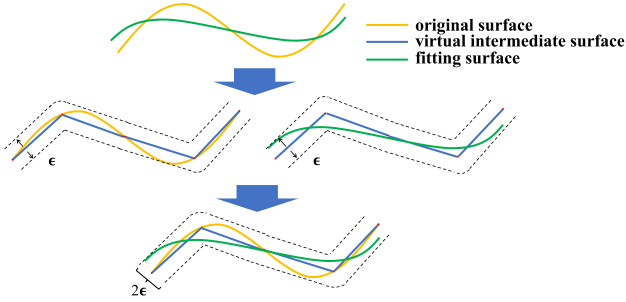}
  \caption{2D illustration of the virtual intermediate surface method. The original surface and approximation surfaces are constrained within the error band (the dashed lines), and therefore a $2\epsilon$ error bound on their deviation.}
  \label{error_source}
\end{figure}

The other major problem in NURBS approximation is $C^2$ continuity. As already noted, the PIA algorithm~\cite{2018_lin_survey_pia} is to be used to carry out the approximation. However, a direct application of PIA does not help here because there are complex geometric constraints involved and they can change the original PIA iteration scheme into a new form. A modified version of PIA, i.e., CPIA, is proposed for this reason. In Sec.~\ref{sec:nurbsfit}, we first derive those geometric constraints regarding $C^2$ continuity and then incorporate them into the PIA algorithm to attain  $C^2$ continuous TPMS approximation. A theoretical convergence proof for CPIA's new iteration scheme is given in Sec.~\ref{sec:proof}.

\subsection{Error-controlled TPMS sampling}
\label{sec:errorcontrol}
As already noted, using a virtual intermediate surface, an error bound of $2\epsilon$ can be attained on the deviation between the original surface and the approximation surface. The essential tasks are then constructing the virtual intermediate surface and controlling the errors from the original surface to the virtual intermediate surface (abbr. as O-to-V) and from the approximation surface to the virtual intermediate surface (abbr. as A-to-V) so that they both are under $\epsilon$. In this work, triangular meshes are chosen to construct virtual intermediate surfaces. Actually, constructing the virtual intermediate surface and controlling O-to-V error can be done in a combined way. This is because the O-to-V error is equivalent to the V-to-O error. We only need to sample the original surface in such a density that the mesh from these sample points has a distance of $\epsilon$ to the original surface. We call this $\epsilon$-density. Similarly, the approximation that gives an A-to-V error of $\epsilon$ is to be called $\epsilon$-approximation. 

\subsubsection{$\epsilon$-density}
As can be seen from the zoomed-in illustration of Fig.~\ref{fig:flip-algorithm}, $\epsilon$-density is to determine the $\Delta$ such that the local shape of the orginal surface $S_O(u,v)$ is bounded by the polyhedron generated through offsetting the triangles $(v_1, v_2, v_3)$ and $(v_1, v_3, v_4)$ by $\epsilon$, where $v_1 = S_O(u_0,v_0)$, $v_2 = S_O(u_0+\Delta,v_0)$, $v_3 = S_O(u_0,v_0+\Delta)$, and $v_4 = S_O(u_0+\Delta,v_0+\Delta)$.

The density parameter $\Delta$ and the offset value $\epsilon$ are related by the following equation (the Fillip's algorithm~\cite{1987_Filip_error_estimation}):
\begin{equation}\label{eq:sampledensity}
    \epsilon = \frac{(M_1+2M_2+M_3)}{8} \Delta^2
\end{equation}
where $M_1$, $M_2$, and $M_3$ are the maximum \rev{second order}{second-order} partial derivatives of $S_O(u,v)$, as follows:
\begin{align*}
    & M_1=\max_{(u,v)\in D}(\left |\frac{\partial ^2x}{\partial u^2}  \right |, \left |\frac{\partial ^2y}{\partial u^2}  \right |, \left |\frac{\partial ^2z}{\partial u^2}  \right |) \\
    & M_2=\max_{(u,v)\in D}(\left |\frac{\partial ^2x}{\partial u\partial v}  \right |, \left |\frac{\partial ^2y}{\partial u\partial v}  \right |, \left |\frac{\partial ^2z}{\partial u\partial v}  \right |) \\
    & M_3=\max_{(u,v)\in D}(\left |\frac{\partial ^2x}{\partial v^2}  \right |, \left |\frac{\partial ^2y}{\partial v^2}  \right |, \left |\frac{\partial ^2z}{\partial v^2}  \right |)
\end{align*}
$D$ is the parameter domain of interest, \rev{}{and the method to estimate the maximum second-order partial derivatives is given in Section 3 of the supplementary material}.

Based on Eq.~\eqref{eq:sampledensity}, the $\epsilon$-density algorithm consists of three major steps:
\begin{enumerate}
    \item Substitute the given error threshold $\epsilon$ to get the sampling density parameter $\Delta$;
    \item Discretize the parameter domain of $S_O(u,v)$ into a grid with size $\Delta$; and
    \item Map the parametric grid to 3D space through $S_O(u,v)$ and then mesh it into the intended virtual intermediate surface.
\end{enumerate}

\begin{figure}[h]
  \centering
  \includegraphics[width=0.48\textwidth]{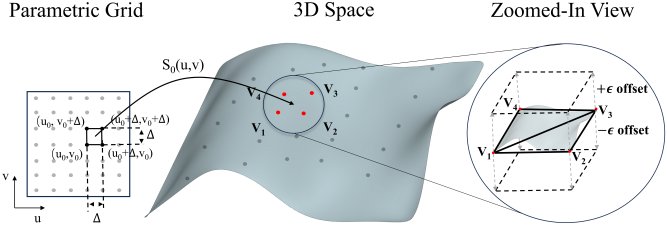}
  \caption{2D illustration of $\epsilon$-density.}
  \label{fig:flip-algorithm}
\end{figure}

\subsubsection{$\epsilon$-approximation}
Because it is hard to estimate the error from the approximation surface to the virtual intermediate surface before approximation, the error is controlled with an indirect approach. That is, the error is estimated after approximation and is utilized to refine the sampling. This refinement increases the sampling density if the error is beyond $ \epsilon$. For error estimation, the algorithm proposed by Zheng et al.~\cite{2000_Zheng_error_estimation_nurbs} is adopted due to its tightness. 
The error estimation equation is:
\begin{small}
\begin{equation}
\begin{aligned}\label{approximation_error}
  \epsilon_{ij} &= \frac{L_{grid}^2n^2 (\mathbf{A}_{ij}+ r\mathbf{G}_{ij})}{4\mathrm{inf}\{ w(u,v) \} +L_{grid}^2n^2\mathbf{G}_{ij}} , 0\le i\le n-1, 0\le j\le m-1\\
 \mathbf{A}_{ij} &=\parallel w_{i,j}\mathbf{P}_{i,j}-w_{i+1,j}\mathbf{P}_{i+1,j}-w_{i,j+1}\mathbf{P}_{i,j+1}+w_{i+1,j+1}\mathbf{P}_{i+1,j+1}\parallel  \\
 \mathbf{G_{ij}} &=\mid w_{i+1,j+1}-w_{i+1,j}-w_{i,j+1}+w_{i,j} \mid
\end{aligned}
\end{equation}
\end{small}
where $ \epsilon_{ij}$ is the estimated tolerance in the $(i,j)$th grid of control points, $L_{grid}$ the distance between each two grid points, $ m$ and $ n$ the number of control points in $ u$ and $ v$ directions plus one, and $r = \max_{0\le i\le n \atop 0\le j\le m}\parallel \mathbf{P}_{ij} \parallel$.

After error estimation, the adaptive refinement of sampling density is carried out by adding a row and a column of sample points locally in the grids where the local error is larger than $\epsilon$. The above refinement is repeated until the approximation error is below $\epsilon$. 

\subsection{$C^2$ continuity-preserving TPMS approximation}
\label{sec:nurbsfit}
Approximating data points ${\mathbf{Q}_{ij}} (i=0,\dots,m_1, j=0,\dots,m_2)$ with a NURBS surface is to solve the following optimization problem~\cite{2021_Flores_TPMS_NURBS_generation_Gyroid}:
\begin{equation}
\min_{{\mathbf{P}_{00},\cdots,\mathbf{P}_{n_1n_2}}} \sum_{i=0}^{m1}\sum_{j=0}^{m2} || \mathbf{Q}_{ij} - \sum_{k=0}^{n1}\sum_{l=0}^{n2} \mathbf{B}(u_i)\mathbf{B}(v_j)w_{kl}\mathbf{P}_{kl} ||^2_2
\end{equation}
where $u_i$ and $v_j$ are the parameter values corresponding to the data point $ \mathbf{Q}_{ij} $. \rev{}{The parameter values are generated using the following procedures. First, the fundamental patch is sampled into point clouds, both in the 2D parameter domain and the 3D space. Then the parameter domain of the fundamental patch is mapped to a triangular region. This region is extended to a rectangular region by the symmetry property, and the same for 2D parametric samples, which results in a quasi-grid of parametric samples. Finally, those parametric samples are associated with the regular grid points in [0,1]×[0,1], correspondingly.}

\rev{This}{The above} optimization problem is a least square problem, and its solution can be obtained by solving the following linear system:
\begin{equation}\label{eq:approximation-linear-equation}
    (\mathbf{B}^T\mathbf{B})\mathbf{P}=\mathbf{B}^T\mathbf{Q}
\end{equation}
where $ \mathbf{P}=[\mathbf{P}_{00},\cdots,\mathbf{P}_{n_1n_2}]^T$, $ \mathbf{Q}=[\mathbf{Q}_{00},\cdots,\mathbf{Q}_{m_1m_2}]^T$, and
\begin{align*}
    \mathbf{B}=
\begin{bmatrix}
\mathbf{B}_{0,p}(u_0)\mathbf{B}_{0,p}(v_0) & \dots & \mathbf{B}_{n_1,p}(u_0)\mathbf{B}_{n_2,p}(v_0)\\
\vdots & \ddots & \vdots\\
\mathbf{B}_{0,p}(u_{m_1})\mathbf{B}_{0,p}(v_{m_2}) & \dots & \mathbf{B}_{n_1,p}(u_{m_1})\mathbf{B}_{n_2,p}(v_{m_2})
\end{bmatrix}
\end{align*}

If a high TPMS translation accuracy is required (i.e., large-scale sampling points to be approximated), the matrix $\mathbf{B}$ in Eq.~\eqref{eq:approximation-linear-equation} is large, and solving this equation is time-consuming. For this reason, the fast PIA algorithm is used instead of solving the equation directly. Its basic steps are summarized below.

The initial step of PIA is constructing a NURBS surface using the data points $\mathbf{Q}$ directly as control points, as follows:
\begin{equation}\label{eq:pia-original}
    \mathbf{P}^0(u,v)=\sum_{i=0}^{m_1} \sum_{j=0}^{m_2} \mathbf{B}_i(u)\mathbf{B}_j(v)w_{ij}\mathbf{P}^0_{ij},0\le u,v\le1
\end{equation}
where $\mathbf{P}^0_{ij} = \mathbf{Q}_{ij}$. The second step is computing the deviation from each ${\mathbf{Q}}_{ij}$ to its associative point on surface $\mathbf{P}^0(u,v)$, as follows:
\begin{equation}\label{eq:pia-original-adjustment}
    \mathbf{\Delta}^0_{ij}=\mathbf{Q}_{ij}-\mathbf{P}^0(u_i,v_j)
\end{equation}
The third step is using this deviation to update the initial control points in Eq.~\eqref{eq:pia-original}:
\begin{equation}\label{eq:pia-original-iteration}
    \mathbf{P}^1_{ij}=\mathbf{P}^0_{ij}+\mathbf{\Delta}^0_{ij}
\end{equation}
Repeating Eqs.~\eqref{eq:pia-original-adjustment} and~\eqref{eq:pia-original-iteration} yields a sequence of surfaces $\{\mathbf{P}^0(u,v), \mathbf{P}^1(u,v), \cdots, \mathbf{P}^{k}(u,v)\}$. The limit of this surface sequence is equivalent to solving Eq.~\eqref{eq:approximation-linear-equation}, as proved by Lin et al.~\cite{2018_lin_survey_pia}. It is easy to see that this iterative scheme is highly parallel, thereby an efficient method.

A direct application of PIA is not enough for TPMS approximation. $C^2$ continuity constraints for smoothly sewing primitive surfaces need to be derived and then incorporated into the PIA algorithm, leading to a new iteration scheme (i.e., CPIA). In the following, we first derive those constraints and then present CPIA's details, followed by the convergence proof of its iteration scheme.

\subsubsection{$C^2$ continuity constraints}
\rev{To preserve $ C^2$ continuity in approximation, different continuity control methods need to be respectively developed for Gyroid, Diamond, and SchwarzP TPMS because their primitive surfaces are different, as shown in Fig.~\ref{surface_symmetry_g2}. They are presented separately in the following paragraphs. }{To preserve $ C^2$ continuity in approximation, geometric constraints need to be added where surfaces meet. To find the junctions of surfaces, we carefully considered the surface assembly methods used in this paper (i.e., the methods in ~\cite{2021_Flores_TPMS_NURBS_generation_Gyroid, 1995_Gandy_TPMS_explicit_representation_Weierstrass_Diamond, 2000_Gandy_TPMS_explicit_representation_Weierstrass_SchwarzP}). As is shown in Fig.~\ref{surface_symmetry_g2}, the assembled offset surfaces are constructed by the copy and rigid transformation of two primitive surfaces with $ offset=+\delta$ and $ offset=-\delta$. The specific rigid transformation matrices for assembly can be found in Eq. (6), Eq. (16), and Eq. (19) of ~\cite{2021_Flores_TPMS_NURBS_generation_Gyroid} for Gyroid, Table 1 in ~\cite{1995_Gandy_TPMS_explicit_representation_Weierstrass_Diamond} for Diamond, and Table 1 in ~\cite{2000_Gandy_TPMS_explicit_representation_Weierstrass_SchwarzP} for Schwarz\_P. (We omit the details of the matrices here because they are not new results. Also, the specific entries in those matrices are not important, nor do they affect the presentation of this work.) Because these surfaces are arranged regularly, the constraints at all the junctions of surfaces can be reduced to four fundamental situations shown in Table~\ref{tab:assemble_edge_connection}, with the number of edges denoted as $ e_1^+, e_2^+, e_3^+, e_4^+$ and $ e_1^-, e_2^-, e_3^-, e_4^-$. Other situations can be obtained by rigid transformation of these four situations.}

\begin{figure}[t]
  \centering  \includegraphics[width=0.45\textwidth]{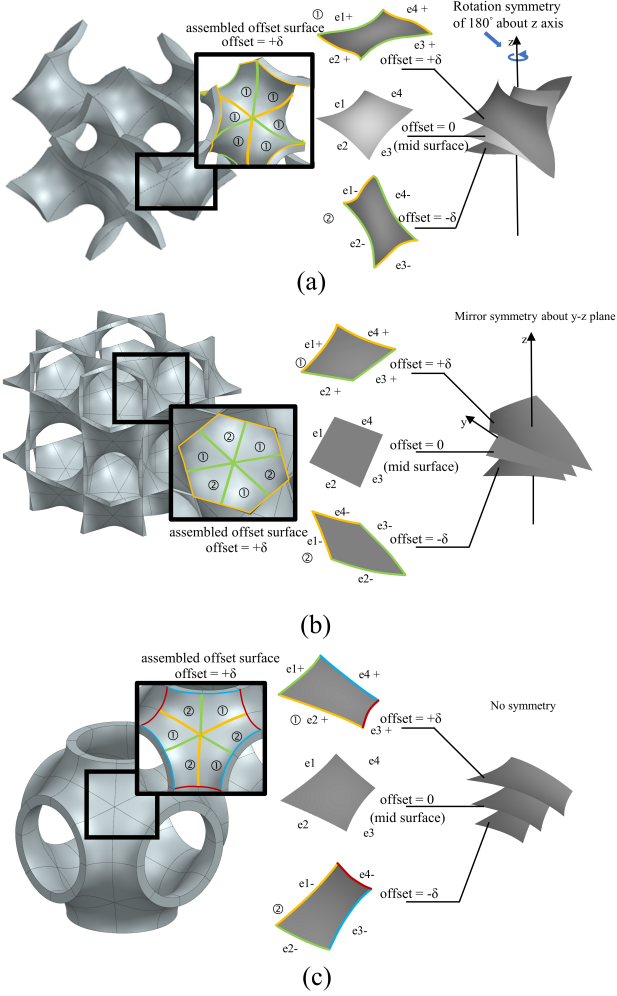}
  \caption{Illustration of the $C^2 $ \rev{continuity preserving}{continuity-preserving} method: (a) Gyroid; (b) Diamond; (c) Schwarz\_P. The edges of the surfaces are divided into two, two, and four types with different colors for Gyroid, Diamond, and Schwarz\_P TPMS respectively. Only the edges with the same color can be spliced together during assembly.}
  \label{surface_symmetry_g2}
\end{figure}
\rev{}{\begin{table}
    \renewcommand{\arraystretch}{1.2}
    \footnotesize
    \setlength{\abovecaptionskip}{0cm}
    \setlength\extrarowheight{3pt}
    \caption{Constraints of the edges of primitive surfaces for Gyroid, Diamond, and Schwarz\_P. $ \mathbf{T}_{ig} (i=1,2)$, $ \mathbf{T}_{id} (i=1,2)$, and $ \mathbf{T}_{ip} (i=1,2,3,4)$ are given in Section 1 of the supplementary material.}
    \centering
    \setlength{\tabcolsep}{3mm} {
    \begin{tabular}{l l}
    \hline
         \multirow{1}{*}{Model Type} & \multirow{1}{*}{Constraints} \\
         \hline
         \multirow{1}{*}{Gyroid} & \multirow{1}{*}{$e_{1}^+ = \mathbf{T}_{1g}e_{2}^-$ \quad $e_{2}^+ = \mathbf{T}_{1g}^{-1}e_{1}^-$ \quad $ e_{3}^+ =\mathbf{T}_{2g}e_{4}^-$ \quad $ e_{4}^+ =\mathbf{T}_{2g}^{-1}e_{3}^-$} \\
         % \multirow{4}*{Gyroid} &  $e_{1}^+ = \mathbf{T}_{1g}e_{2}^-$  \\
         % ~ &  $e_{2}^+ = \mathbf{T}_{1g}^{-1}e_{1}^-$ \\
         % ~ &  $ e_{3}^+ =\mathbf{T}_{2g}e_{4}^-$ \\
         % ~ &  $ e_{4}^+ =\mathbf{T}_{2g}^{-1}e_{3}^-$ \\
         \hline
         \multirow{1}{*}{Diamond} & \multirow{1}{*}{$e_{1}^+ = \mathbf{T}_{1d}e_{4}^-$ \quad $e_{2}^+ = \mathbf{T}_{2d}e_{3}^-$ \quad $e_{3}^+ = \mathbf{T}_{2d}^{-1}e_{2}^-$ \quad $e_{4}^+ = \mathbf{T}_{1d}^{-1}e_{1}^-$} \\
         % \multirow{4}*{Diamond} & $e_{1}^+ = \mathbf{T}_{1d}e_{4}^-$  \\
         % ~ & $e_{2}^+ = \mathbf{T}_{2d}e_{3}^-$ \\
         % ~ & $e_{3}^+ = \mathbf{T}_{2d}^{-1}e_{2}^-$ \\
         % ~ & $e_{4}^+ = \mathbf{T}_{1d}^{-1}e_{1}^-$ \\
         \hline
         \multirow{1}{*}{Schwarz\_P} & \multirow{1}{*}{$e_{1}^+ = \mathbf{T}_{1p}e_{1}^+$ \quad $e_{2}^+ = \mathbf{T}_{2p}e_{2}^+$ \quad $e_{3}^+ = \mathbf{T}_{3p}e_{3}^+$ \quad $e_{4}^+ = \mathbf{T}_{4p}e_{4}^+$} \\
         % \multirow{4}*{Schwarz\_P} & $e_{1}^+ = \mathbf{T}_{1p}e_{2}^-$  \\
         % ~ &  $e_{2}^+ = \mathbf{T}_{2p}e_{1}^-$ \\
         % ~ &  $e_{3}^+ = \mathbf{T}_{3p}e_{4}^-$  \\
         % ~ &  $e_{4}^+ = \mathbf{T}_{4p}e_{3}^-$ \\
         \hline
    \end{tabular} }
    \label{tab:assemble_edge_connection}
\end{table}}

\rev{}{For each constraint in Table~\ref{tab:assemble_edge_connection}, a constraint equation is derived to adjust the corresponding control points of the approximation surface to keep the first-order and second-order derivatives at the junctions to be the same~\cite{2017_Hu_G2_continuity_control_points_constraints}. Because the edge-edge connection relationships for the three types of TPMS are different, their constraints are different and thus considered separately.}

\textbf{Gyroid.} 
% For Gyroid TPMS, the two primitive surfaces used to assemble the offset surface have the same shape (shown in Fig.~\ref{surface_symmetry_g2}a). That is, the primitive surface with $ offset=-\delta$ can be obtained by taking a specific rigid transformation, denoted as $\mathbf{T}_{3g}$, to the primitive surface with $ offset=+\delta$. $\mathbf{T}_{3g}$ is given in Section 1 of the supplementary material. Based on this observation, the splicing of two primitive surfaces can be considered as splicing one primitive surface with another primitive surface copied and transformed by itself, and the edge-edge connection relationships for different edges of the single primitive surface are derived from the edge-edge connection relationships for edges of two primitive surfaces in Table ~\ref{tab:assemble_edge_connection}, given as:
\rev{As is shown in Fig.~\ref{surface_symmetry_g2}a, the assembled offset surface is constructed by the copy and transformation of a single primitive surface with $ offset=+\delta$. The specific rigid transformation matrices for assembly can be found in Table 1 in ~\cite{2021_Flores_TPMS_NURBS_generation_Gyroid}. (We omit the details of the matrices here because they are not new results. Also, the specific entries in those matrices are not important, nor do they affect the presentation of this work.) To ensure $C^2 $ continuity when assembling these transformed surfaces, constraints need to be added where surfaces meet.}{}
\rev{To derive the continuity constraints, we use two connected surfaces, denoted as $S_1$ and $S_2$, as an example. The four edges of S1 are denoted as ${e_1,e_2,e_3,e_4}$, and the four edges of S2 are denoted as ${e\prime_1,e\prime_2,e\prime_3,e\prime_4}$, see Fig.~\ref{surface_symmetry_g2}a for their arrangement. For Gyroid, edge $e_1$ is connected with edge $ e\prime_2$. Since $S_2$ is a rigid transformed version of $S_1$, i.e., $S_2=\mathbf{T}S_1$, we have $ e\prime_2=\mathbf{T}e_2$. Then there exists a relationship between $e_1$ and $e_2$:}{}
\begin{equation*}
    e_1=e\prime_2=\mathbf{T}e_2
\end{equation*}
\rev{Generally, to control the surface continuity, constraints should be added to $e_1$ and $e\prime_2$, which belong to two different surfaces. But in our situation, the constraint, i.e., e1=Te2, is actually added within one primitive surface. Based on this observation, constraints are added between the edges within one primitive surface for $C^2$ continuity preservation. Based on this observation, constraints are added between the edges within one primitive
surface for $C^2$ continuity preservation.}{}
\rev{To preserve $C^2$ continuity, the first-order and second-order derivatives of the surfaces at the junctions need to be kept the
 same[27]. Noticing that Gyroid TPMS has an intrinsic rotation symmetry (shown in Fig.5a), a tougher constraint is utilized to simplify the control approach. That is, the intrinsic rotation symmetry is maintained in the approximation surface. Although this constraint is tougher since it brings in constraints for more control points than the general method, the calculation is simplified because the constraint matrix shows a high degree
 of symmetry. Specifically, the constraint takes the following form:}{}
\begin{equation*}
    \mathbf{M_gPA}=\mathbf{P}
\end{equation*}
\rev{where $\mathbf{P}$ is the vector of control points. The matrix $\mathbf{M_g}$ and $\mathbf{A}$ are utilized to map the coordinates of the control points to the corresponding position under the rotation symmetry mentioned before. Their derivations are complex, and we have put them in Section 1 of the supplementary material to save pages.}{}
\rev{}{For Gyroid TPMS, the two primitive surfaces used to assemble the offset surface have the same shape (shown in Fig.~\ref{surface_symmetry_g2}a). That is, the primitive surface with $ offset=-\delta$ can be obtained by taking a specific rigid transformation, denoted as $\mathbf{T}_{3g}$, to the primitive surface with $ offset=+\delta$. $\mathbf{T}_{3g}$ is given in Section 1 of the supplementary material. Based on this observation, the splicing of two primitive surfaces can be considered as splicing one primitive surface with another primitive surface copied and transformed by itself, and the constraints for different edges of the single primitive surface are derived from the constraints for edges of two primitive surfaces in Table~\ref{tab:assemble_edge_connection}, given as:}
\begin{small}
    \begin{equation}\label{gyroid_edge_edge_relation}
   e_{1}^+ = \mathbf{T}_{4g}e_{1}^+, \quad  e_{2}^+ = \mathbf{T}_{5g}e_{2}^+, \quad e_{3}^+ = \mathbf{T}_{6g}e_{3}^+, \quad e_{4}^+ = \mathbf{T}_{7g}e_{4}^+
\end{equation}
\end{small}
\rev{}{where $ \mathbf{T}_{4g}$-$ \mathbf{T}_{7g}$ are derived from $ \mathbf{T}_{1g}$-$ \mathbf{T}_{3g}$ and they are given in Section 1 of the supplementary material.}
% where $ \mathbf{T}_{4g}$-$ \mathbf{T}_{7g}$ are derived from $ \mathbf{T}_{1g}$-$ \mathbf{T}_{3g}$ and they are given in Section 1 of the supplementary material.

\rev{}{To ensure $C^2$ continuity at the junctions, part of the control points of the approximation surface that control the first-order and second-order derivatives of the two connected edges are extracted and adjusted for consistency.
Specifically, for the four constraints in Eq.~\ref{gyroid_edge_edge_relation}, if the control points of the approximation surface are denoted as $ \mathbf{P}$ and the group of matrices used to extract the corresponding control points are denoted as $ \mathbf{M}_{ijg} (i=1,2,3,4,5$ and $j=1,2,3,4)$, then the constraint equations are derived as:}
% To ensure $C^2$ continuity for the edge-edge connections, part of the control points of the approximation surface that control the first-order and second-order derivatives of the two connected edges are extracted and adjusted for consistency.
% Specifically, for the four edge-edge connection relationships in Eq.~\ref{gyroid_edge_edge_relation}, if the control points of the approximation surface are denoted as $ \mathbf{P}$ and the group of matrices used to extract part of the control points are denoted as $ \mathbf{M}_{ijg} (i=1,2,3,4,5$ and $j=1,2,3,4)$, then the constraint equations are derived as:
\begin{small}
\begin{align}\label{eq:constraint_g}
\begin{aligned}
    %     &(\mathbf{M}_{21g}-\mathbf{M}_{11g})\mathbf{P}\mathbf{T}_{1g}\mathbf{T}_{3g}=(\mathbf{M}_{21g}-\mathbf{M}_{11g}) \mathbf{P} \\
    % & (3\mathbf{M}_{41g}-2\mathbf{M}_{31g}-\mathbf{M}_{51g}) \mathbf{P}\mathbf{T}_{1g}\mathbf{T}_{3g} =(3\mathbf{M}_{41g}-2\mathbf{M}_{31g} -\mathbf{M}_{51g})\ \mathbf{P}\\
    %     &(\mathbf{M}_{22g}-\mathbf{M}_{12g})\mathbf{P}\mathbf{T}_{1g}^{-1}\mathbf{T}_{3g}^{-1}=(\mathbf{M}_{22g}-\mathbf{M}_{12g}) \mathbf{P} \\
    % & (3\mathbf{M}_{42g}-2\mathbf{M}_{32g}-\mathbf{M}_{52g}) \mathbf{P}\mathbf{T}_{1g}^{-1}\mathbf{T}_{3g}^{-1} =(3\mathbf{M}_{42g}-2\mathbf{M}_{32g} -\mathbf{M}_{52g})\ \mathbf{P}\\
    % & (\mathbf{M}_{23g}-\mathbf{M}_{13g})\mathbf{P}\mathbf{T}_{2g}\mathbf{T}_{3g}=(\mathbf{M}_{23g}-\mathbf{M}_{13g})\mathbf{P} \\
    % &(3\mathbf{M}_{43g}-2\mathbf{M}_{33g}-\mathbf{M}_{53g})\mathbf{P}\mathbf{T}_{2g}\mathbf{T}_{3g}=(3\mathbf{M}_{43g}-2\mathbf{M}_{33g}-\mathbf{M}_{53g}) \mathbf{P}\\
    % & (\mathbf{M}_{24g}-\mathbf{M}_{14g})\mathbf{P}\mathbf{T}_{2g}^{-1}\mathbf{T}_{3g}^{-1}=(\mathbf{M}_{24g}-\mathbf{M}_{14g})\mathbf{P} \\
    % &(3\mathbf{M}_{44g}-2\mathbf{M}_{34g}-\mathbf{M}_{54g})\mathbf{P}\mathbf{T}_{2g}^{-1}\mathbf{T}_{3g}^{-1}=(3\mathbf{M}_{44g}-2\mathbf{M}_{34g}-\mathbf{M}_{54g}) \mathbf{P} \\
    &(\mathbf{M}_{2ig}-\mathbf{M}_{1ig})\mathbf{P}\mathbf{T}_{(i+3)g}=(\mathbf{M}_{2ig}-\mathbf{M}_{1ig}) \mathbf{P} \\
    & (3\mathbf{M}_{4ig}-2\mathbf{M}_{3ig}-\mathbf{M}_{5ig}) \mathbf{P}\mathbf{T}_{(i+3)g}=(3\mathbf{M}_{4ig}-2\mathbf{M}_{3ig} -\mathbf{M}_{5ig})\ \mathbf{P}\\
    & i=1,2,3,4
\end{aligned}
\end{align}
\end{small}
\rev{}{where the left side of the equation indicates the first-order and second-order derivatives of one edge and the right side of the equation indicates the first-order or second-order derivatives of the other edge, and the same coefficients on both sides of the equations to calculate the first-order or second-order derivatives of NURBS surfaces are eliminated to simplify the form. $\mathbf{M}_{ijg}$ $ (i=1,2,3,4,5$ and $ j=1,2,3,4)$ are used to extract the first, second, and third rows of control points at the four edges of the approximation surface. Specifically, $(\mathbf{M}_{2ig}-\mathbf{M}_{1ig})\mathbf{P}$ $ (i=1,2,3,4)$ are utilized to calculate the first-order derivative of the surface, and $ (3\mathbf{M}_{4ig}-2\mathbf{M}_{3ig}-\mathbf{M}_{5ig})\mathbf{P}$ $ (i=1,2,3,4)$ are utilized to calculate the second-order derivative of the surface. Their details are provided in Section 1 of the supplementary material.}
% where the left side of the equation indicates the first-order or second-order derivatives of one edge and the right side of the equation indicates the first-order or second-order derivatives of the other edge, and the same coefficients on both sides of the equations to calculate the first-order or second-order derivatives of NURBS surfaces are eliminated to simplify the form.

\textbf{Diamond.} 
% For Diamond TPMS, the two primitive surfaces for assembling have different shapes (shown in Fig.~\ref{surface_symmetry_g2}b).
% To ensure $c^2$ continuity with the same approach for Gyroid, the continuity for each edge-edge connection relationship (shown in Table ~\ref{tab:assemble_edge_connection}) is considered during the approximation of these two primitive surfaces.
% Similar to Gyroid, the $C^2$ continuity constraint equations are derived as:
\rev{For Diamond TPMS, the assembled offset surface is constructed by the copy and transformation of two different primitive surfaces with $ offset=+\delta$ and $ offset=-\delta$ (shown in Fig.~\ref{surface_symmetry_g2}b). Again, take two connected surfaces $S_1$ and $S_2$ as an example. $S_1$ has four edges $ e_1$, $ e_2$, $ e_3$, and $e_4$, and $S_2$ has four edges $ e_1^{\prime} $, $ e_2^{\prime} $, $ e_3^{\prime} $, and $e_4^{\prime} $. For Diamond, edge $ e_1$ of $S_1$ is connected with edge $ e_1^{\prime}$ of $S_2$, then $ e_1$ could be expressed as the transformation of $ e_1^{\prime}$ with a rigid transformation matrix denoted as $\mathbf{T}$. That is, $ e_1=\mathbf{T}e_1^{\prime}$. Applying this reasoning to all edges (i.e., $e_{1}^+$, $e_{2}^+$, $e_{3}^+$, $e_{4}^+$, $e_{1}^-$, $e_{2}^-$, $e_{3}^-$, $e_{4}^-$) gives the following edge-edge relationships:}{}
\begin{equation*}
   e_{1}^+ = \mathbf{T}_1e_{1}^-, \quad  e_{2}^+ = \mathbf{T}_2e_{2}^-, \quad e_{3}^+ = \mathbf{T}_3e_{3}^-, \quad e_{4}^+ = \mathbf{T}_4e_{4}^-
\end{equation*}
\rev{where $\mathbf{T}_1$-$ \mathbf{T}_4$ are the rigid transformation matrices given in Section 1 of the supplementary material.}{}

\rev{Similar to the Gyroid case, $C^2$ continuous constraints are imposed on the first-order and second-order derivatives of $S_1$ and $S_2$. Specifically, the constraints take the following form:}{}
\begin{align*}
\begin{aligned}
        &(\mathbf{M}_{2i}-\mathbf{M}_{1i})\mathbf{N}_2 \mathbf{P}\mathbf{A}_i^{-1}=(\mathbf{M}_{2i}-\mathbf{M}_{1i})\mathbf{N}_1 \mathbf{P} \\
    & (3\mathbf{M}_{4i}-2\mathbf{M}_{3i}-\mathbf{M}_{5i})\mathbf{N}_2 \mathbf{P}\mathbf{A}_i^{-1} =(3\mathbf{M}_{4i}-2\mathbf{M}_{3i} -\mathbf{M}_{5i})\mathbf{N}_1 \mathbf{P}\\
    & (\mathbf{M}_{2i}-\mathbf{M}_{1i})\mathbf{N}_1 \mathbf{P}\mathbf{A}_i=(\mathbf{M}_{2i}-\mathbf{M}_{1i})\mathbf{N}_2 \mathbf{P} \\
    &(3\mathbf{M}_{4i}-2\mathbf{M}_{3i}-\mathbf{M}_{5i})\mathbf{N}_1 \mathbf{P}\mathbf{A}_i=(3\mathbf{M}_{4i}-2\mathbf{M}_{3i}-\mathbf{M}_{5i})\mathbf{N}_2 \mathbf{P}\\
    &i=1,2,3,4
\end{aligned}
\end{align*}
\rev{where $\mathbf{P} =\begin{pmatrix}
 \mathbf{P}_1\\
\mathbf{P}_2
\end{pmatrix}$. $\mathbf{P}_1$ and $\mathbf{P}_2$ are the homogeneous coordinates of control points for the two approximation surfaces. $ \mathbf{M}_{1i}$-$\mathbf{M}_{5i} $ $ (i=1,2,3,4)$ are used to extract the first, second, and third rows of control points at the four edges of the approximation surface. $\mathbf{M}_{2i}-\mathbf{M}_{1i}$ $ (i=1,2,3,4)$ are utilized to calculate the first-order derivative of the surface, and $ 3\mathbf{M}_{4i}-2\mathbf{M}_{3i}-\mathbf{M}_{5i}$ $ (i=1,2,3,4)$ are utilized to calculate the second-order derivative of the surface. $ \mathbf{N}_1$ and $\mathbf{N}_2$ are used to assemble the control points $\mathbf{P}_1$ and $ \mathbf{P}_2$ to $\mathbf{P}$. Matrices $\mathbf{A}_i $ $ (i=1,2,3,4)$ and their inverse matrices are the rigid transformation matrices. Their details can be found in Section 1 of the supplementary material.}{}

\rev{}{For Diamond TPMS, the two primitive surfaces for assembling have different shapes (shown in Fig.~\ref{surface_symmetry_g2}b).
To ensure $C^2$ continuity with the same approach for Gyroid, the continuity for each constraint (shown in Table~\ref{tab:assemble_edge_connection}) is considered during the approximation of these two primitive surfaces.
Similar to Gyroid, the $C^2$ continuity constraint equations are derived as:}
\begin{small}
\begin{align}\label{eq:constraint_d}
\begin{aligned}
        &(\mathbf{M}_{24d}-\mathbf{M}_{14d})\mathbf{N}_{i} \mathbf{P}\mathbf{T}_{1d}=(\mathbf{M}_{21d}-\mathbf{M}_{11d})\mathbf{N}_{3-i} \mathbf{P} \\
    & (3\mathbf{M}_{44d}-2\mathbf{M}_{34d}-\mathbf{M}_{54d})\mathbf{N}_{i} \mathbf{P}\mathbf{T}_{1d} =(3\mathbf{M}_{41d}-2\mathbf{M}_{31d} -\mathbf{M}_{51d})\mathbf{N}_{3-i} \mathbf{P} \\
    & (\mathbf{M}_{23d}-\mathbf{M}_{13d})\mathbf{N}_{i} \mathbf{P}\mathbf{T}_{2d}=(\mathbf{M}_{22d}-\mathbf{M}_{12d})\mathbf{N}_{3-i} \mathbf{P} \\
    &(3\mathbf{M}_{43d}-2\mathbf{M}_{33d}-\mathbf{M}_{53d})\mathbf{N}_{i} \mathbf{P}\mathbf{T}_{2d}=(3\mathbf{M}_{42d}-2\mathbf{M}_{32d}-\mathbf{M}_{52d})\mathbf{N}_{3-i} \mathbf{P} \\
    & (i=1,2)
\end{aligned}
\end{align}
\end{small}
\rev{}{where $\mathbf{P} =\begin{pmatrix}
 \mathbf{P}_1\\
\mathbf{P}_2
\end{pmatrix}$. $\mathbf{P}_1$ and $\mathbf{P}_2$ are the homogeneous coordinates of control points for the two approximation surfaces. $ \mathbf{M}_{ijd}$ $ (i=1,2,3,4,5$ and $ j=1,2,3,4)$ have the same function as $ \mathbf{M}_{ijg}$ $ (i=1,2,3,4,5$ and $ j=1,2,3,4)$ for Gyroid. $ \mathbf{N}_1$ and $\mathbf{N}_2$ are used to assemble the control points $\mathbf{P}_1$ and $ \mathbf{P}_2$ to $\mathbf{P}$. Specifically, $ \mathbf{N}_1\mathbf{P}=\begin{pmatrix}
 \mathbf{P}_1\\
\mathbf{0}
\end{pmatrix}$ and $ \mathbf{N}_2\mathbf{P}=\begin{pmatrix}
 \mathbf{P}_2\\
\mathbf{0}
\end{pmatrix}$. Matrices $\mathbf{T}_{id} $ $ (i=1,2)$ are the rigid transformation matrices given in Table~\ref{tab:assemble_edge_connection}. Their details can be found in Section 1 of the supplementary material.}
% where $\mathbf{P} =\begin{pmatrix}
%  \mathbf{P}_1\\
% \mathbf{P}_2
% \end{pmatrix}$. $\mathbf{P}_1$ and $\mathbf{P}_2$ are the homogeneous coordinates of control points for the two approximation surfaces. $ \mathbf{M}_{ijd}$ $ (i=1,2,3,4,5$ and $ j=1,2,3,4)$ are used to extract the first, second, and third rows of control points at the four edges of the approximation surface. Specifically, $(\mathbf{M}_{2id}-\mathbf{M}_{1id})$ $ (i=1,2,3,4)$ are utilized to calculate the first-order derivative of the surface, and $ (3\mathbf{M}_{4id}-2\mathbf{M}_{3id}-\mathbf{M}_{5id})$ $ (i=1,2,3,4)$ are utilized to calculate the second-order derivative of the surface. $ \mathbf{N}_1$ and $\mathbf{N}_2$ are used to assemble the control points $\mathbf{P}_1$ and $ \mathbf{P}_2$ to $\mathbf{P}$. Matrices $\mathbf{T}_{id} $ $ (i=1,2)$ are the rigid transformation matrices given in Table ~\ref{tab:assemble_edge_connection}. Their details can be found in Section 1 of the supplementary material.

\textbf{Schwarz\_P.} 
% For Schwarz\_P TPMS, only one primitive surface is needed to assemble the offset surface(shown in Fig.~\ref{surface_symmetry_g2}c). That is, the offset surface with $offset=+\delta$ can be obtained by copying and transforming the primitive surface with $offset=+\delta$.
% The same approach to ensure $C^2$ continuity for Gyroid is used here for each of the edge-edge connection relationships in Table ~\ref{tab:assemble_edge_connection} and the constraint equations are derived as:
\rev{For Schwarz\_P TPMS, the two primitive surfaces used to assemble the offset surface have the same shape (shown in Fig.~\ref{surface_symmetry_g2}c). Following the same derivation procedures as in the Diamond case, the edge-edge relationships are:}{}
\begin{equation*}
   e_{1}^+ = \mathbf{T}_5e_{2}^-, \quad  e_{2}^+ = \mathbf{T}_6e_{1}^-, \quad e_{3}^+ = \mathbf{T}_7e_{4}^-, \quad e_{4}^+ = \mathbf{T}_8e_{3}^-
\end{equation*}
\rev{where $ \mathbf{T}_5$-$ \mathbf{T}_8$ are the rigid transformation matrices given in Section 1 of the supplementary material. Edges $e_{1}^+$, $e_{2}^+$, $e_{3}^+$, $e_{4}^+$, $e_{1}^-$, $e_{2}^-$, $e_{3}^-$, $e_{4}^-$ are shown in Fig.~\ref{surface_symmetry_g2}c. Imposing the above constraints on the first-order and second-order derivatives of the approximation surfaces gives the following $C^2$ continuity constraints:}{}
\begin{align*}
\begin{aligned}
    &  (\mathbf{M}_{2i}-\mathbf{M}_{1i})\mathbf{P}\mathbf{A}_i=(\mathbf{M}_{2i}-\mathbf{M}_{1i})\mathbf{P} \\
&  (3\mathbf{M}_{4i}-2\mathbf{M}_{3i}-\mathbf{M}_{5i})\mathbf{P}\mathbf{A}_i=(3\mathbf{M}_{4i}-2\mathbf{M}_{3i}-\mathbf{M}_{5i})\mathbf{P} \\
& i=1,2,3,4
\end{aligned}
\end{align*}
\rev{where $ \mathbf{P}$ is the vector of control points. $ \mathbf{M}_{1i}$-$ \mathbf{M}_{5i} $ $ (i=1,2,3,4)$ are the same matrices as in Diamond. $ \mathbf{A}_i \ (i=1,2,3,4)$ are rigid transformation matrices, whose details have been put in Section 1 of the supplementary material.}{}
\rev{}{For Schwarz\_P TPMS, only one primitive surface is needed to assemble the offset surface (shown in Fig.~\ref{surface_symmetry_g2}c). That is, the offset surface with $offset=-\delta$ can be obtained by copying and transforming the primitive surface with $offset=+\delta$.
The same approach to ensure $C^2$ continuity for Gyroid is used here for each of the constraints in Table~\ref{tab:assemble_edge_connection} and the constraint equations are derived as:}
\begin{small}
    \begin{align}\label{eq:constraint_p}
\begin{aligned}
    &  (\mathbf{M}_{2ip}-\mathbf{M}_{1ip})\mathbf{P}\mathbf{T}_{ip}=(\mathbf{M}_{2ip}-\mathbf{M}_{1ip})\mathbf{P} \\
&  (3\mathbf{M}_{4ip}-2\mathbf{M}_{3ip}-\mathbf{M}_{5ip})\mathbf{P}\mathbf{T}_{ip}=(3\mathbf{M}_{4ip}-2\mathbf{M}_{3ip}-\mathbf{M}_{5ip})\mathbf{P} \\
& i=1,2,3,4
\end{aligned}
\end{align}
\end{small}
\rev{}{where $ \mathbf{P}$ is the vector of control points. $ \mathbf{M}_{ijp}$ $ (i=1,2,3,4,5$ and $ j=1,2,3,4)$ are the same as $ \mathbf{M}_{ijd}$ $ (i=1,2,3,4,5$ and $ j=1,2,3,4)$ for Diamond. Their details have been put in Section 1 of the supplementary material.}
% where $ \mathbf{P}$ is the vector of control points. $ \mathbf{M}_{ijp}$ $ (i=1,2,3,4,5$ and $ j=1,2,3,4)$ have the same function as $ \mathbf{M}_{ijd}$ $ (i=1,2,3,4,5$ and $ j=1,2,3,4)$ for Diamond. Their details have been put in Section 1 of the supplementary material.

\rev{}{For $C^2$ continuity-preserving, a general problem is that when there are extraordinary points, it is hard to achieve high-order parametric continuity. However, the three types of TPMS structures considered in this paper are intrinsic $C^2$ continuous, effectively avoiding the issue of extraordinary points.}

\subsubsection{The CPIA method}
CPIA has the same major steps as PIA, but with a different iterative scheme. The first step of CPIA is constructing a NURBS surface utilizing the sample points $\mathbf{Q}$ as control points, denoted as $\mathbf{P}^0(u_i,v_j)$. The second step is computing the update term $\Delta$ for $\mathbf{P}^0(u_i,v_j)$. In PIA, this term refers to the deviation from each sample point ${\mathbf{Q}}_{ij}$ to its associative point on the approximation surface $\mathbf{P}^0(u,v)$. However, in CPIA, an additional constraint term should also be considered to ensure that the approximation surface satisfies the continuity constraints in Eq.~\eqref{eq:constraint_g}, Eq.~\eqref{eq:constraint_d}, and Eq.~\eqref{eq:constraint_p}.

For Gyroid, this term is derived from Eq.~\eqref{eq:constraint_g} by calculating the difference between the two sides of the equation. This difference is actually the deviation between the theoretical value that meets the constraints and the actual value. Then a factor $\frac{1}{2}$ is added to distribute the deviation into two equal parts since the deviation is caused by two sides of the surfaces at their junction. The difference is added to the original PIA update term (i.e., Eq.~\eqref{eq:pia-original-adjustment}) to obtain the new CPIA update term:
\begin{small}
\begin{align}
    \begin{aligned}\label{eq:cpia-g-update}
    \mathbf{\Delta}^0_{ij}&=\mathbf{Q}_{ij}-\mathbf{P}^0(u_i,v_j) \\
    & +\frac{1}{2}\sum_{i=1}^{4}((\mathbf{M}_{2ig}-\mathbf{M}_{1ig})\mathbf{P}^0(u_i,v_j)\mathbf{T}_{(i+3)g}-(\mathbf{M}_{2ig}-\mathbf{M}_{1ig}) \mathbf{P}^0(u_i,v_j))\\
    & +\frac{1}{2}\sum_{i=1}^{4}((3\mathbf{M}_{4ig}-2\mathbf{M}_{3ig}-\mathbf{M}_{5ig}) \mathbf{P}^0(u_i,v_j)\mathbf{T}_{(i+3)g}\\
    &-(3\mathbf{M}_{4ig}-2\mathbf{M}_{3ig} -\mathbf{M}_{5ig})\ \mathbf{P}^0(u_i,v_j))
    % (\mathbf{M_g}\mathbf{P}^0(u_i,v_j)\mathbf{A}-\mathbf{P}^0(u_i,v_j))
\end{aligned}
\end{align}
\end{small}

For Diamond, the additional constraint term is derived from Eq.~\eqref{eq:constraint_d} with similar reasoning: calculating the difference between the two sides of the equation and multiplying it with a factor $\frac{1}{2}$, and then adding it to Eq.~\eqref{eq:pia-original-adjustment}, yielding:
\begin{small}
\begin{align}
\begin{aligned}
            \mathbf{\Delta}^0_{ij}&=\mathbf{Q}_{ij}-\mathbf{P}^0(u_i,v_j)\\
        &+\frac{1}{2}\sum_{i=1}^{2} ((\mathbf{M}_{24d}-\mathbf{M}_{14d})\mathbf{N}_{i} \mathbf{P}^0(u_i,v_j)\mathbf{T}_{1d}-(\mathbf{M}_{21d}-\mathbf{M}_{11d})\mathbf{N}_{3-i} \mathbf{P}^0(u_i,v_j)) \\
    &+\frac{1}{2}\sum_{i=1}^{2} ((3\mathbf{M}_{44d}-2\mathbf{M}_{34d}-\mathbf{M}_{54d})\mathbf{N}_{i} \mathbf{P}^0(u_i,v_j)\mathbf{T}_{1d} \\
    &-(3\mathbf{M}_{41d}-2\mathbf{M}_{31d} -\mathbf{M}_{51d})\mathbf{N}_{3-i} \mathbf{P}^0(u_i,v_j))\\
&+\frac{1}{2}\sum_{i=1}^{2} ((\mathbf{M}_{23d}-\mathbf{M}_{13d})\mathbf{N}_{i} \mathbf{P}^0(u_i,v_j)\mathbf{T}_{2d}-(\mathbf{M}_{22d}-\mathbf{M}_{12d})\mathbf{N}_{3-i} \mathbf{P}^0(u_i,v_j)) \\
    &+\frac{1}{2}\sum_{i=1}^{2}((3\mathbf{M}_{43d}-2\mathbf{M}_{33d}-\mathbf{M}_{53d})\mathbf{N}_{i} \mathbf{P}^0(u_i,v_j)\mathbf{T}_{2d} \\
    &-(3\mathbf{M}_{42d}-2\mathbf{M}_{32d}-\mathbf{M}_{52d})\mathbf{N}_{3-i} \mathbf{P}^0(u_i,v_j))
\end{aligned}
\end{align}
\end{small}

For Schwarz\_P, the update term follows directly as:
\begin{small}
\begin{align}
\begin{aligned}
        \mathbf{\Delta}^0_{ij}&=\mathbf{Q}_{ij}-\mathbf{P}^0(u_i,v_j)\\
    &+\frac{1}{2} \sum_{i=1}^{4} ((\mathbf{M}_{2ip}-\mathbf{M}_{1ip})\mathbf{P}^0(u_i,v_j)\mathbf{T}_{ip}-(\mathbf{M}_{2ip}-\mathbf{M}_{1ip})\mathbf{P}^0(u_i,v_j)) \\
     &+\frac{1}{2}\sum_{i=1}^{4}((3\mathbf{M}_{4ip}-2\mathbf{M}_{3ip}-\mathbf{M}_{5ip})\mathbf{P}^0(u_i,v_j)\mathbf{T}_{ip} \\
     &-(3\mathbf{M}_{4ip}-2\mathbf{M}_{3ip}-\mathbf{M}_{5ip})\mathbf{P}^0(u_i,v_j))
\end{aligned}
\end{align}
\end{small}

Having obtained the update terms, the third step utilizes them to update the control points:
\begin{equation}\label{eq:cpia-original-iteration}
    \mathbf{P}^1_{ij}=\mathbf{P}^0_{ij}+\mathbf{\Delta}^0_{ij}
\end{equation}

Eqs.~\eqref{eq:cpia-g-update}-\eqref{eq:cpia-original-iteration} will be repeated until the approximation error is below a given tolerance. Algorithm~\ref{algo:cpia} gives more details on these repetitive procedures. 

\begin{algorithm}[h]
    \caption{The CPIA Algorithm}
    \label{algo:cpia}
    \begin{algorithmic}[1]
        \Require An ordered point set $\{\mathbf{Q}_{ij}\}^{m_1,m_2}_{i=0,j=0}$ and corresponding parameters $\{\xi_i,\zeta_j\}^{m_1,m_2}_{i=0,j=0}$, weight of control points $w_{ij}$, knot vectors $U=\{0=u_0<u_1<\cdots<u_n=1\}$ and $V=\{0=v_0<v_1<\cdots<v_n=1\}$, tolerance $\epsilon$, constraint term $ \mathbf{C}$
        \Ensure a set of control points $\{\mathbf{P}_{kl}\}_{k=0,l=0}^{n_1,n_2}$ interpolating the point set $\{\mathbf{Q}_{ij}\}^{m_1,m_2}_{i=0,j=0}$
    \State $\{\mathbf{P}_{kl}\}_{k=0,l=0}^{n_1,n_2} \leftarrow \{\mathbf{Q}_{ij}\}^{m_1,m_2}_{i=0,j=0}$ \Comment{Initialize the control points with sample points}
    \State $FP \leftarrow $EvaluateFittedSurfacePoints($ \{\mathbf{P}_{kl}\}_{k=0,l=0}^{n_1,n_2}, U,V,$
    \Statex $ \{\xi_i,\zeta_j\}^{m_1,m_2}_{i=0,j=0},w_{ij}$)
    \State $\delta_{ij} \leftarrow$ AddConstraints($\{\mathbf{P}_{kl}\}_{k=0,l=0}^{n_1,n_2} ,\mathbf{C} $)
    \State $\Delta_{ij} \leftarrow$ CalculateDeviation($ FP, \{\mathbf{Q}_{ij}\}^{m_1,m_2}_{i=0,j=0}, \delta_{ij}$)
    \While{$\Delta_{ij} > \epsilon $}
        \State $\{\mathbf{P}_{kl}\}_{k=0,l=0}^{n_1,n_2} \leftarrow$ RegenerateControlPoints($ \Delta_{ij}$)
        \State $FP \leftarrow $EvaluateFittedSurfacePoints($ \{\mathbf{P}_{kl}\}_{k=0,l=0}^{n_1,n_2}, U,V, $
        \Statex $\{\xi_i,\zeta_j\}^{m_1,m_2}_{i=0,j=0},w_{ij}$)
        \State $\delta_{ij} \leftarrow$ AddConstraints($\{\mathbf{P}_{kl}\}_{k=0,l=0}^{n_1,n_2} ,\mathbf{C} $)
        \State $\Delta_{ij} \leftarrow$ CalculateDeviation($ FP, \{\mathbf{Q}_{ij}\}^{m_1,m_2}_{i=0,j=0}, \delta_{ij}$)
    \EndWhile
    \end{algorithmic}
\end{algorithm}

\subsection{CPIA convergence proof}
\label{sec:proof}
Considering the page limit, only the convergence proof for Gyroid CPIA is given here. Those for Diamond and Schwarz\_P basically follow the same proof procedures but with different equations. We have put them in Section 2 of the supplementary material.

\begin{proposition}
    The CPIA method for Gyroid is convergent and the limit surface is the least-square approximation outcome of the initial data $\mathrm{\{\mathbf{Q}_{ij}\}^{m_1,m_2}_{i=0,j=0}}$.
\end{proposition}
\begin{proof}
Using CPIA, a sequence of control points and offset surfaces is generated. To show its convergence, let $
\mathbf{P}^k=\{\mathbf{P}_0^k,\mathbf{P}_1^k,\cdots,\mathbf{P}_n^k\}^T$ and $\mathbf{Q}=\{\mathbf{Q}_0,\mathbf{Q}_1,\cdots,\mathbf{Q}_m \}^T$ denote the control points at $k$-th iteration and the sample points, respectively. In the $(k+1)$th iteration, we have:
\begin{small}
    \begin{align}\label{eq:iterative-scheme-1}
        \begin{aligned}
        \mathbf{P}^{k+1} = &\mathbf{P}^k+(\mathbf{Q-BwP}^k) \\
        &+\frac{1}{2}\sum_{i=1}^{4}((\mathbf{M}_{2ig}-\mathbf{M}_{1ig})\mathbf{P}^k\mathbf{T}_{(i+3)g}-(\mathbf{M}_{2ig}-\mathbf{M}_{1ig}) \mathbf{P}^k) \\
        &+\frac{1}{2}\sum_{i=1}^{4}((3\mathbf{M}_{4ig}-2\mathbf{M}_{3ig}-\mathbf{M}_{5ig}) \mathbf{P}^k\mathbf{T}_{(i+3)g} \\
        &-(3\mathbf{M}_{4ig}-2\mathbf{M}_{3ig} -\mathbf{M}_{5ig})\ \mathbf{P}^k) \\
        \end{aligned}
    \end{align}
\end{small}
\rev{The matrices $ \mathbf{M_g}$ and $\mathbf{A}$ are the same as those in Sec.~\ref{sec:nurbsfit}, and they are orthogonal matrices, i.e., $\mathbf{M_g}^{2}=\mathbf{I}$ and $\mathbf{A}^{2}=\mathbf{I}$. Let $\mathbf{r}^k=\mathbf{P}^k-\mathbf{B}^{-1}\mathbf{Q}$, $ \ \mathbf{D}=\frac{\mathbf{I}}{2} -\mathbf{B}$. Then Eq. (20) can be rewritten as:}{}
\begin{small}
    \begin{align*}
        \begin{aligned}
                &\mathbf{r}^{k+1} =  \mathbf{D}\mathbf{r}^k+\frac{1}{2}\mathbf{M_g}\mathbf{r}^k\mathbf{A} \\
            & = \mathbf{D}(\mathbf{D}\mathbf{r}^{k-1}+\frac{\mathbf{M_g}}{2}\mathbf{r}^{k-1}\mathbf{A})+\frac{\mathbf{M_g}}{2}(\mathbf{D}\mathbf{r}^{k-1}+\frac{\mathbf{M_g}}{2}\mathbf{r}^{k-1}\mathbf{A})\mathbf{A} \\
            & =\mathbf{D}^2\mathbf{r}^{k-1}+2\mathbf{D}\frac{\mathbf{M_g}}{2}\mathbf{r}^{k-1}\mathbf{A}+(\frac{1}{2})^2\mathbf{r}^{k-1}\\
            & = \cdots \\
            & = \sum_{i=0}^{k+1} { k+1\choose i}\mathbf{D}^i(\frac{\mathbf{M_g}}{2})^{k+1-i}\mathbf{r}^0\mathbf{A}^{k+1-i}
        \end{aligned}
    \end{align*}
\end{small}
\rev{Suppose $\{ \lambda_i(\mathbf{M_g}) \}(i=0,1,\cdots,m^2-1)$ and $\{ \lambda_j(\mathbf{A}) \}(j=0,1,2)$ are the eigenvalues of $ \mathbf{M_g}$ and $ \mathbf{A}$ sorted in non-decreasing order. Since $\left | \lambda_i(\mathbf{M_g}) \right |=1$ and $\left | \lambda_i(\mathbf{A}) \right |=1$, their powers have no effect on convergence. Let}{}
\begin{equation*}
    \mathbf{s}^{k+1}=\sum_{i=0}^{k+1} { k+1\choose i}\mathbf{D}^i(\frac{\mathbf{M_g}}{2})^{k+1-i}\mathbf{r}^0\mathbf{A}^{k+1-i}
\end{equation*}
\rev{The convergence of $\mathbf{r}^{k+1}$ is thus the same as $\mathbf{s}^{k+1}$. Then we only need to show that $\{ \mathbf{s}^{k}\}$ is convergent. With the Binomial theorem, we get:}{}
\begin{align*}
    \mathbf{s}^{k+1}=&(\mathbf{D}+\frac{\mathbf{I}}{2})^{k+1}\mathbf{r}^0 \\
    &=(\mathbf{I}-\mathbf{B})^{k+1}\mathbf{r}^0
\end{align*}
\rev{From Theorem 2.2 in [40], we know that $\rho (\mathbf{I}-\mathbf{B})< 1$, where $\rho (\mathbf{I}-\mathbf{B})$ is the spectral radius of $\mathbf{I}-\mathbf{B}$. Therefore, we have the following equation:}{}
\begin{align*}
\lim_{k \to \infty} (\mathbf{I}-\mathbf{B})^k= \lim_{k \to \infty} \mathbf{s}^{k+1} =(\mathbf{0})_{n+1}
\end{align*}
\rev{ It is equivalent to:}{}
\begin{align*}
\lim_{k \to \infty} \mathbf{P}^k-\mathbf{B}^{-1}\mathbf{Q}=(\mathbf{0})_{n+1}
\end{align*}

\rev{}{The matrices $ \mathbf{T}_{ig}$ $ (i=4,5,6,7)$ are mentioned in Sec.~\ref{sec:nurbsfit}, and they are invertible matrices where the absolute values of the eigenvalues are all 1, i.e., $\left | \lambda_i ( \mathbf{T}_{jg})\right| = 1, i=1,2,3,4$ and $ j=4,5,6,7$. The matrices $ \mathbf{M}_{ijg}$ $ (i=1,2,3,4,5$ and $ j=1,2,3,4)$ are also given in Sec.~\ref{sec:nurbsfit}. $\mathbf{B}$ is the B-spline basis function matrix and $\mathbf{w}$ is the weight matrix. They are invertible. Let $\mathbf{r}^k=\mathbf{P}^k-\mathbf{w}^{-1}\mathbf{B}^{-1}\mathbf{Q}$, $ \ \mathbf{D}=\mathbf{I} -\mathbf{Bw}$, $\alpha_i =3\mathbf{M}_{4ig}-2\mathbf{M}_{3ig}-\mathbf{M}_{5ig}$, and $\beta_i = \mathbf{M}_{2ig}-\mathbf{M}_{1ig} (i=1,2,3,4)$. Then Eq.~\eqref{eq:iterative-scheme-1} can be rewritten as:}
% The matrices $ \mathbf{T}_{ig}$ $ (i=4,5,6,7)$ are mentioned in Sec.~\ref{sec:nurbsfit}, and they are invertible matrices where the absolute values of the eigenvalues are all 1, i.e., $\left | \lambda_i ( \mathbf{T}_{jg})\right| = 1, i=1,2,3,4$ and $ j=1,2$. The matrices $ \mathbf{M}_{ijg}$ $ (i=1,2,3,4,5$ and $ j=1,2,3,4)$ are also given in Sec.~\ref{sec:nurbsfit}. $\mathbf{B}$ is the B-spline basis function matrix and $\mathbf{w}$ is the weight matrix. They are invertible. Let $\mathbf{r}^k=\mathbf{P}^k-\mathbf{w}^{-1}\mathbf{B}^{-1}\mathbf{Q}$, $ \ \mathbf{D}=\mathbf{I} -\mathbf{Bw}$, $\alpha_i =3\mathbf{M}_{4ig}-2\mathbf{M}_{3ig}-\mathbf{M}_{5ig}$, and $\beta_i = \mathbf{M}_{2ig}-\mathbf{M}_{1ig} (i=1,2,3,4)$. Then Eq.~\eqref{eq:iterative-scheme-1} can be rewritten as:
\begin{small}
    \begin{align}
        \begin{aligned}
                &\mathbf{r}^{k+1} =  (\mathbf{D}-\frac{1}{2}\sum_{i=0}^{4}(\alpha_i+\beta_i))\mathbf{r}^k+\frac{1}{2}\sum_{i=0}^{4}(\alpha_i+\beta_i) \mathbf{r}^k\mathbf{T}_{(i+3)g} \\
            & = (\mathbf{D}-\frac{1}{2}\sum_{i=0}^{4}(\alpha_i+\beta_i))((\mathbf{D}-\frac{1}{2}\sum_{i=0}^{4}(\alpha_i+\beta_i))\mathbf{r}^{k-1}+\frac{1}{2}\sum_{i=0}^{4}(\alpha_i+\beta_i) \mathbf{r}^{k-1}\mathbf{T}_{(i+3)g}) \\
            &+\frac{1}{2}\sum_{i=0}^{4}(\alpha_i+\beta_i) ((\mathbf{D}-\frac{1}{2}\sum_{n=0}^{4}(\alpha_n+\beta_n))\mathbf{r}^{k-1}+\frac{1}{2}\sum_{n=0}^{4}(\alpha_n+\beta_n) \mathbf{r}^{k-1}\mathbf{T}_{(n+3)g}) \mathbf{T}_{(i+3)g} \\
            & = \cdots \\
            & = \sum_{i=0}^{k+1} { k+1\choose i}(\mathbf{D}-\frac{1}{2}\sum_{n=0}^{4}(\alpha_n+\beta_n))^i(\frac{1}{2}\sum_{n=0}^{4}(\alpha_n+\beta_n))^{k+1-i}\mathbf{r}^0\mathbf{T}_{(i+3)g}^{k+1-i}
        \end{aligned}
    \end{align}
\end{small}
\rev{}{Supposing $\{ \lambda_k(\alpha_i) \}(k=0,1,\cdots,n-1)$ and $\{ \lambda_l(\beta_j) \}(l=0,1,\cdots,n-1)$ are the eigenvalues of $ \alpha_i$ and $ \beta_j$ sorted in non-decreasing order. n is the number of control points of the approximation surface in the u and v directions. Since $\mid \lambda_k(\alpha_i) \mid=0\quad or \quad 1,k=0,1,\cdots,n-1$, $\mid \lambda_l(\beta_j) \mid=0\quad or \quad 1,l=0,1,\cdots,n-1$, and $ r(\sum_{i=0}^{4}\beta_i)<n$, their powers have no effect on convergence. $r(\sum_{i=0}^{4}\beta_i)$ is the rank of $ (\sum_{i=0}^{4}\beta_i)$. Let}
% Supposing $\{ \lambda_k(\alpha_i) \}(k=0,1,\cdots,n-1)$ and $\{ \lambda_l(\beta_j) \}(l=0,1,\cdots,n-1)$ are the eigenvalues of $ \alpha_i$ and $ \beta_i$ sorted in non-decreasing order. n is the number of control points of the approximation surface in the u and v directions. Since $\mid \lambda_k(\alpha_i) \mid=0\quad or \quad 1,k=0,1,\cdots,n-1$, $\mid \lambda_l(\beta_j) \mid=0\quad or \quad 1,l=0,1,\cdots,n-1$, and $ r(\sum_{i=0}^{4}\beta_i)<n$, their powers have no effect on convergence. $r(\sum_{i=0}^{4}\beta_i)$ is the rank of $ (\sum_{i=0}^{4}\beta_i)$. Let
\begin{small}
    \begin{align}
        \mathbf{s}^{k+1}=\sum_{i=0}^{k+1} { k+1\choose i}(\mathbf{D}-\frac{1}{2}\sum_{n=0}^{4}(\alpha_n+\beta_n))^i(\frac{1}{2}\sum_{n=0}^{4}(\alpha_n+\beta_n))^{k+1-i}\mathbf{r}^0
    \end{align}
\end{small}
\rev{}{To show that $\{ \mathbf{r}^{k}\}$ is convergent, we firstly show that $\{ \mathbf{s}^{k}\}$ is convergent. With the Binomial theorem, we get:}
% The convergence of $\mathbf{r}^{k+1}$ is the same as $\mathbf{s}^{k+1}$ and the proof is given in Section 3 of the supplementary material. Then we only need to show that $\{ \mathbf{s}^{k}\}$ is convergent. With the Binomial theorem, we get:
\begin{small}
\begin{equation}
    \mathbf{s}^{k+1}= \mathbf{D}^{k+1} \mathbf{r}^0=(\mathbf{I}-\mathbf{Bw})^{k+1} \mathbf{r}^0
\end{equation}
\end{small}
\rev{}{From Theorem 2.2 in ~\cite{2005_Lin_pia_for_NURBS}, we know that $\rho (\mathbf{I}-\mathbf{B})<1$, where $\rho (\cdot)$ gets a matrix's spectral radius. With uniform weight assignment, $\rho (\mathbf{I}-\mathbf{Bw})<1$. Therefore, we have the following equation:}
\begin{small}
\begin{align}
\lim_{k \to \infty} (\mathbf{I}-\mathbf{Bw})^k= \lim_{k \to \infty} \mathbf{s}^{k} =(\mathbf{0})_{n+1}
\end{align}
\end{small}
\rev{}{Let}
\begin{small}
\begin{align}
\begin{aligned}
    \mathbf{r}^{k+1}&= \sum_{i=0}^{k+1}\mathbf{r}_i=\sum_{i=0}^{k+1} { k+1\choose i}(\mathbf{D}-\frac{1}{2}\sum_{n=0}^{4}(\alpha_n+\beta_n))^i(\frac{1}{2}\sum_{n=0}^{4}(\alpha_n+\beta_n))^{k+1-i}\mathbf{r}^0\mathbf{T}_{(i+3)g}^{k+1-i} \\
    \mathbf{s}^{k+1}&=\sum_{i=0}^{k+1}\mathbf{s}_i=\sum_{i=0}^{k+1} { k+1\choose i}(\mathbf{D}-\frac{1}{2}\sum_{n=0}^{4}(\alpha_n+\beta_n))^i(\frac{1}{2}\sum_{n=0}^{4}(\alpha_n+\beta_n))^{k+1-i}\mathbf{r}^0 \\
\end{aligned}
\end{align}
\end{small}
\rev{}{For each term $\mathbf{r}_i$ in $\mathbf{r}^{k+1}$ and $\mathbf{s}_i$ in $s^{k+1}$, $ || \mathbf{r}_i ||_2= ||\mathbf{s}_i ||_2 $. $ || \cdot||_2$ gets 2-norm of a matrix. Then $|| \mathbf{r}^{k+1} ||_2 = || \mathbf{s}^{k+1} ||_2 $ since $\mathbf{T}_{ig} (i=4,5,6,7)$ are rigid transformation matrices with modulo lengths $1$ for all the eigenvalues. Because $||\mathbf{r}^{k+1}||_2=\sqrt[]{\lambda _{max}((\mathbf{r}^{k+1})^T \mathbf{r}^{k+1})} $ and $||\mathbf{s}^{k+1}||_2=\sqrt[]{\lambda _{max}((\mathbf{s}^{k+1})^T \mathbf{s}^{k+1})} $, the spectral radius of $\mathbf{r}^{k+1}$ and $\mathbf{s}^{k+1}$ are the same. So if $\{ \mathbf{s}^{k}\}$ is convergent, then $\{ \mathbf{r}^{k}\}$ is also convergent, which is equivalent to:}
\begin{small}
\begin{align}
\lim_{k \to \infty} \mathbf{P}^k-\mathbf{w}^{-1}\mathbf{B}^{-1}\mathbf{Q}=\mathbf{0}
\end{align}
\end{small}
So \{ $\mathbf{P}^k$ \} is convergent, and
\begin{small}
\begin{align}
\mathbf{P}^ \infty=\mathbf{w}^{-1}\mathbf{B}^{-1}\mathbf{Q}
\end{align}
\end{small}
which concludes the proof. 
\end{proof}

\subsection{STEP file generation}
\label{sec:assembling}
As is shown in Fig.~\ref{step_format}, a STEP file is divided into two sections. The first is the head and tail of the file, which contains the meta-data like the author, the file name, and the description of the file. The second contains all the data of the file, including geometry information (e.g., ``CARTESIAN\_POINT"), topological information (e.g., ``CLOSED\_SHELL"), and other parameters such as the tolerance of the model (e.g., ``UNCERTAINTY\_MEASURE\_WITH\_UNIT"). In this region, Express language is used. Express is a formal language designed for storing and transmitting information about engineering data and processes in computer systems. In this paper, the TPMS solid model data is stored in STEP files with this language.

\begin{figure}[h]
  \centering
  \includegraphics[width=0.45\textwidth]{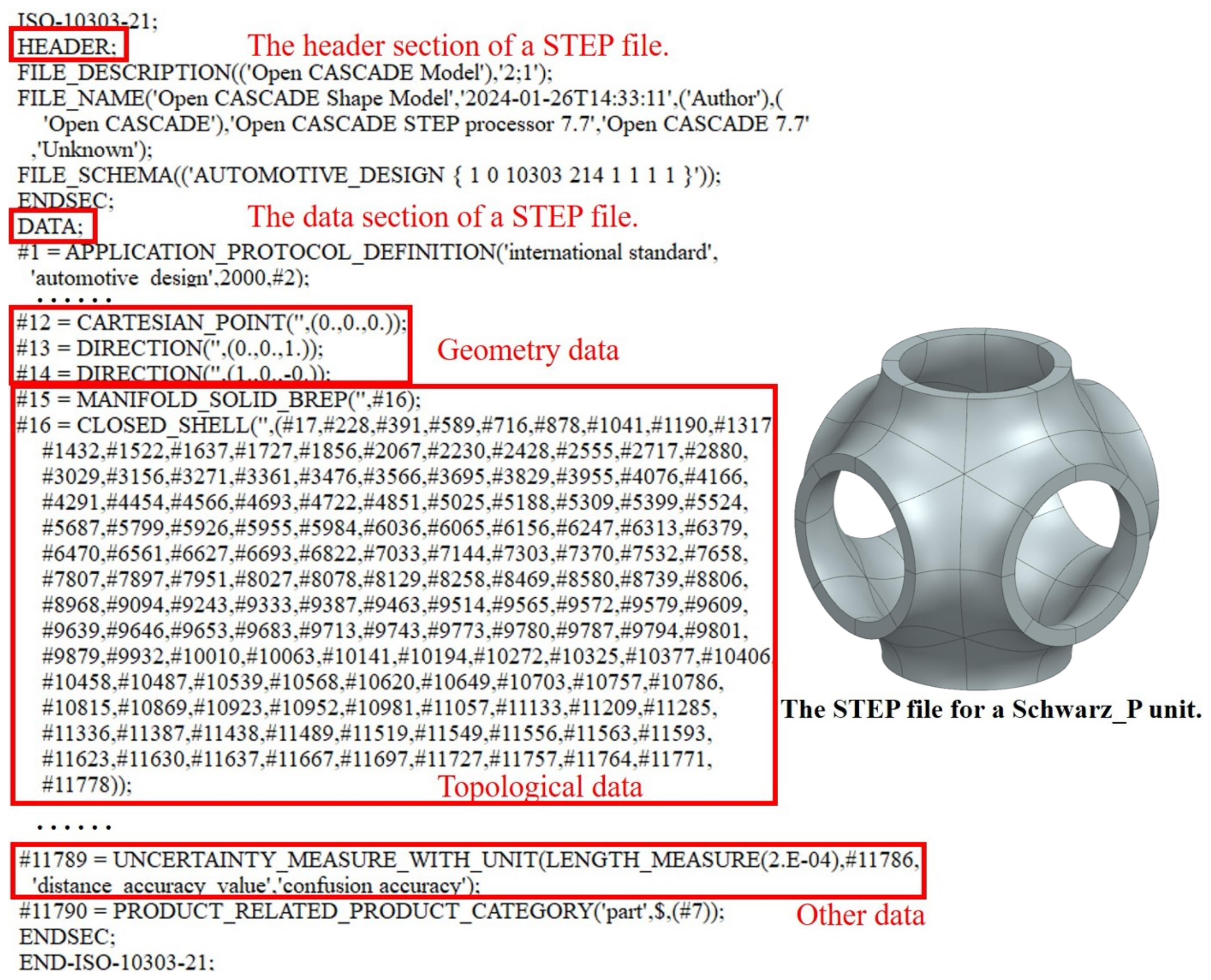}
  \caption{An example of the structure of a STEP file. The structure of a STEP file includes two regions: the head and tail region and the data region. The geometry data, topological data, and other data in data regions are framed in the graph.}
  \label{step_format}
\end{figure}
\begin{figure}[t]
  \centering
  \includegraphics[width=0.35\textwidth]{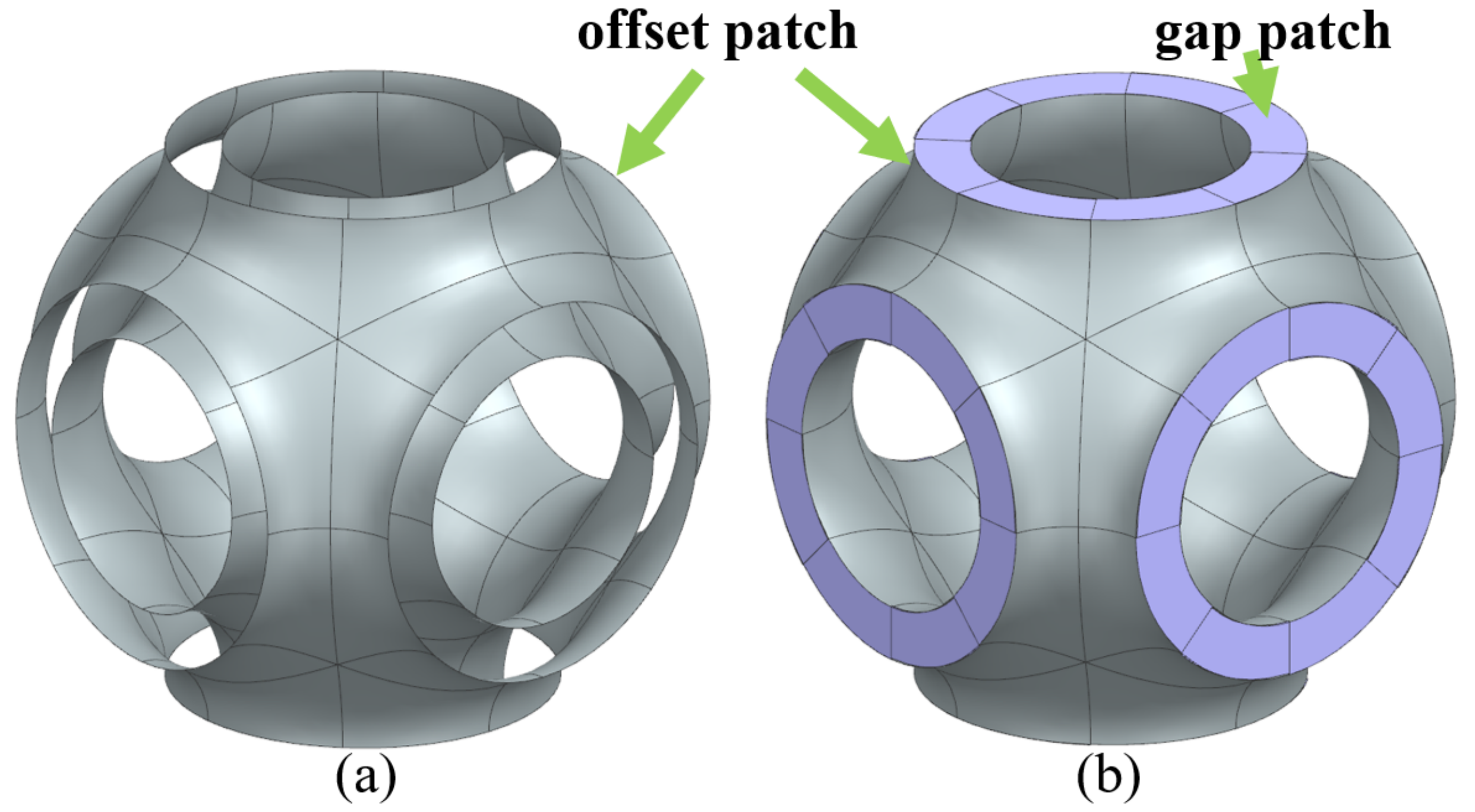}
  \caption{Illustration of gap patches and offset patches: (a) Copied and transformed offset patches without generating gap patches; and (b) Offset patches and corresponding gap patches.}
  \label{gappatch}
\end{figure}

The STEP file generation process used in this work is a standard one and consists of three stages. First, meta-data is written. Second, the NURBS patches are copied, transformed, and assembled into a solid model. Third, the data of the solid model is translated into the Express language and saved in STEP files.

The above three stages are quite straightforward, except for one special task before assembly in the second stage. That is, gap patches need to be crafted to bridge the space between a pair of offset surfaces, as shown in Fig.~\ref{gappatch}. To fill the gaps, we first identify pairs of opposite edges on offset surfaces and then generate surfaces connecting those edge pairs by using the bridging operation offered by existing CAD software.

\begin{figure*}[t]
  \centering
  \includegraphics[width=0.9\textwidth]{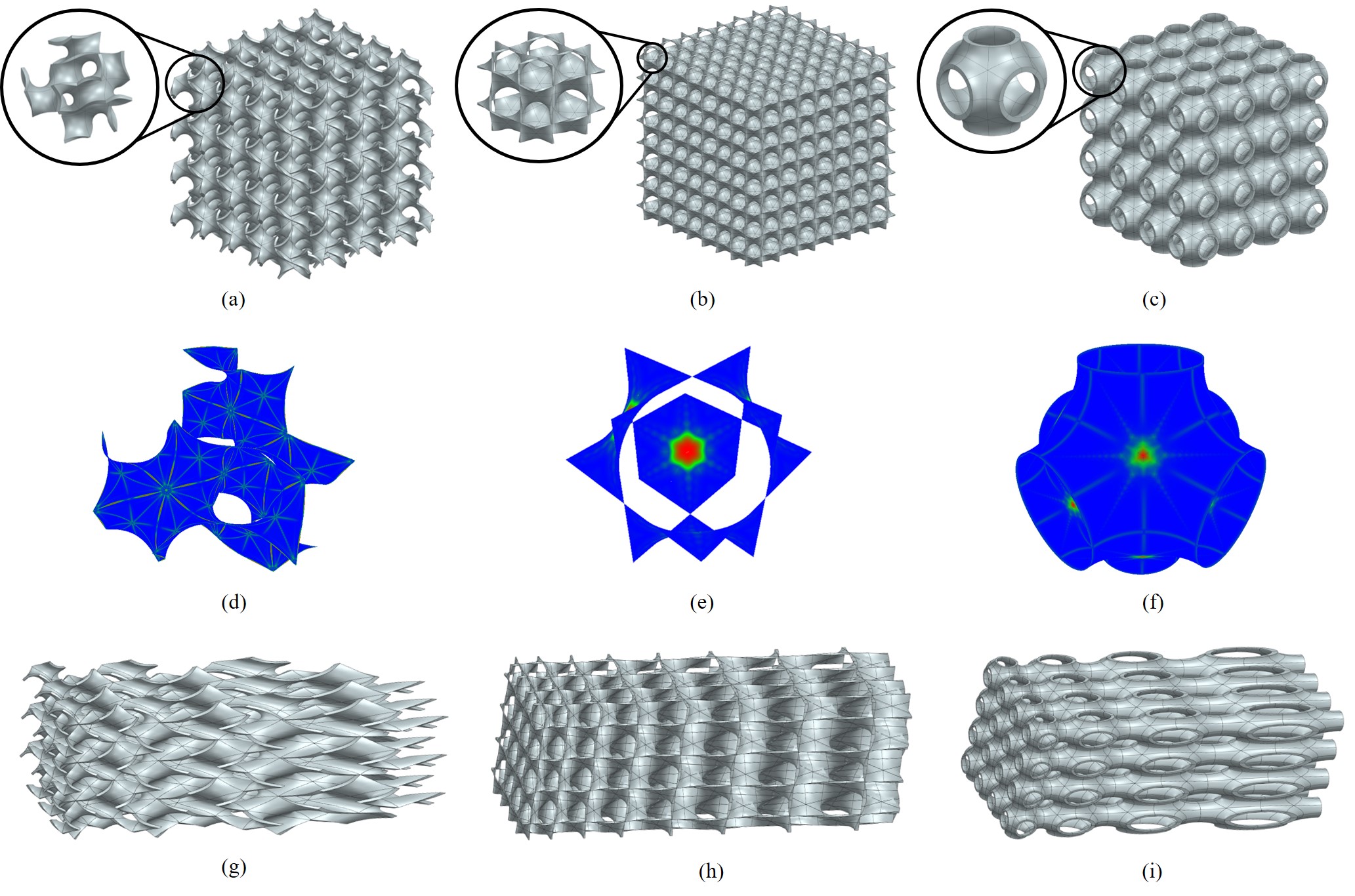}
  \caption{Gyroid (left), Diamond (middle), and Schwarz\_P (right) models and analysis results for case study (1)-(6). (a)-(c) unit models and assembled models; (d)-(f) variation of the second-order partial derivative; (g)-(i) scaled models.}
  \label{result_tpms_structure}
\end{figure*}

\begin{figure}
  \centering
  \includegraphics[width=0.45\textwidth]{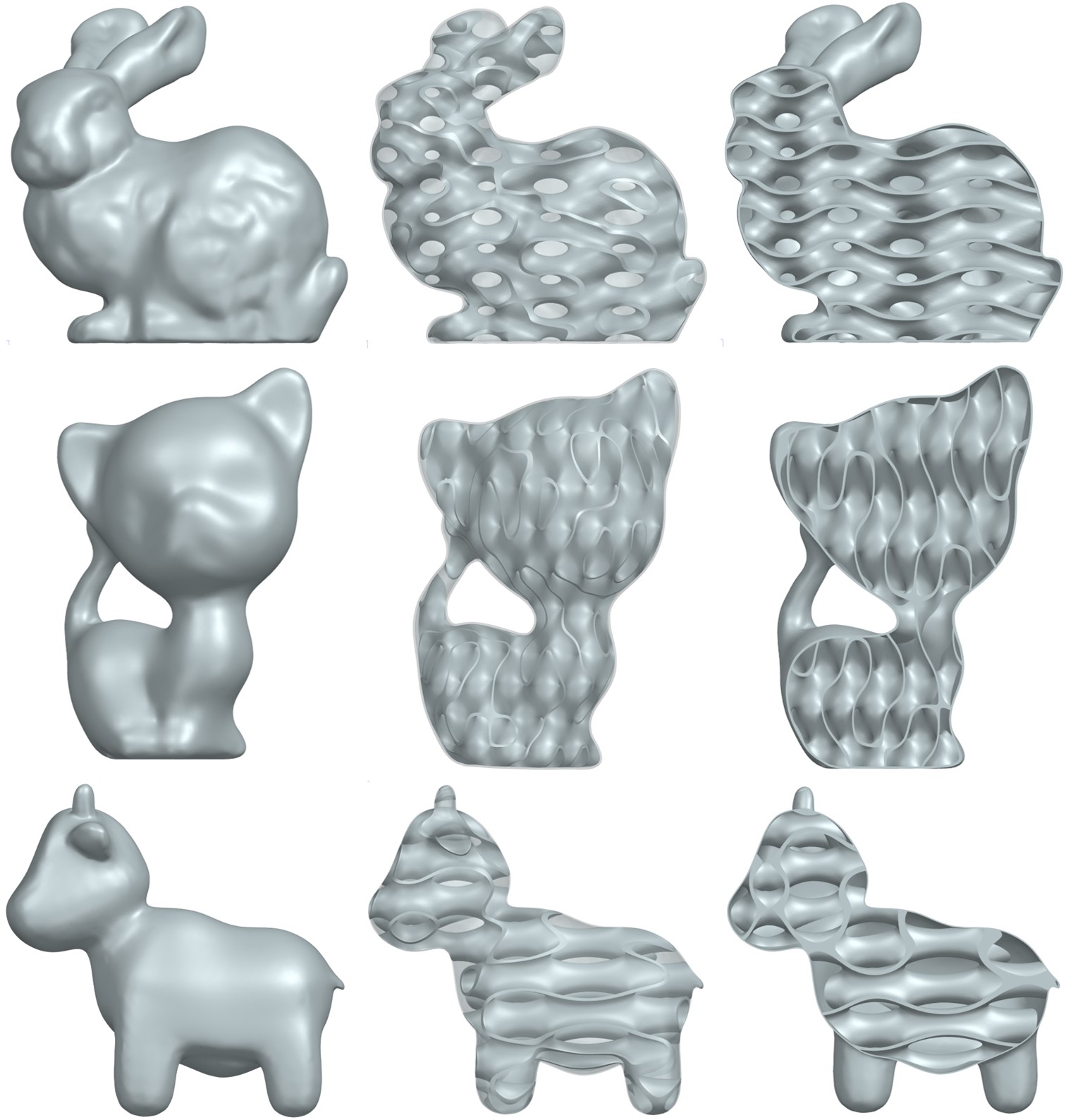}
  \caption{Solid models (left), transparent models (middle), and the section of the models (right) for case study (7)-(9): the top row for models designed with Gyroid;  the middle row for models designed with Diamond; and the bottom row for models designed with Schwarz\_P.}
  \label{result_tpms_structure_intersection}
\end{figure}
\begin{figure*}
  \centering
  \includegraphics[width=0.9\textwidth]{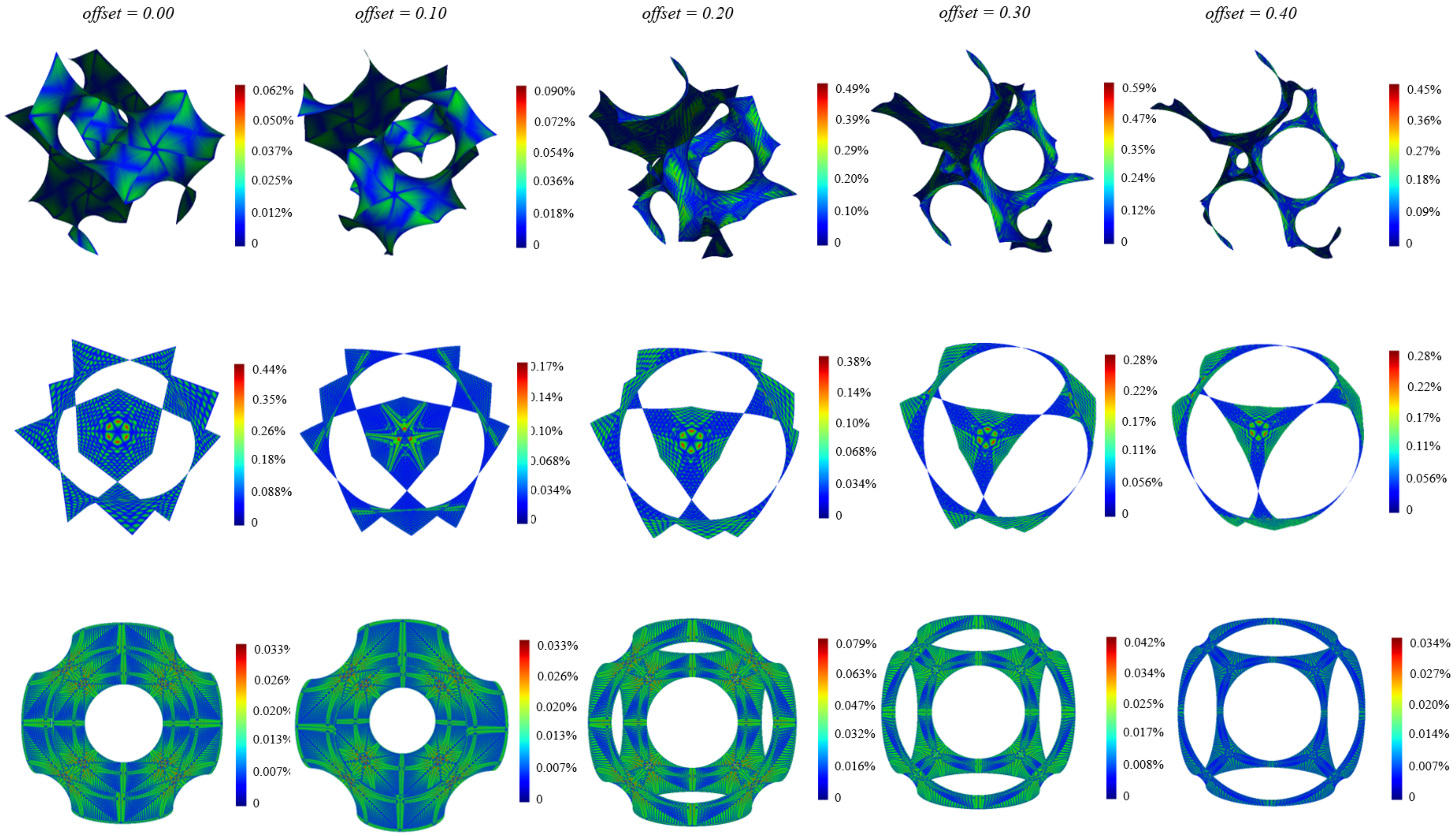}
  \caption{Error distribution in the NURBS surfaces with five different offset values for error analysis: the top row for Gyroid; the middle row for Diamond; and the bottom for Schwarz\_P.}
  \label{result_error_analysis_1}
\end{figure*}

With gap patches in place, they are copied and arranged alongside corresponding offset patches according to the intrinsic symmetry of TPMS. \rev{}{In this process, if we first create a unit (with gaps filled already) and then duplicate it alongside the given directions, overlapped gap patches will cause incorrect model topology. To avoid this problem, we opted to not fill the gaps immediately after creating a unit, but to keep this unit open. Then the open unit is duplicated alongside the given directions, resulting in two shells. The gap between these two shells is then filled. This way of working can avoid overlapping gap patches inside the solid model, and thus there is no need for deleting gap patches.} 

Then all the patches (including offset patches and gap patches) are assembled to form a TPMS unit (shown in Fig.~\ref{gappatch}).
Due to the spatially periodic nature of TPMS, the assignment of spatial periods enables precise transformation and assembly of the patches into a complete solid model. 
Furthermore, because most surface patches share the same geometry and they only differ in their positions and orientations, we only need to store geometry once and all topological data can be concisely stored by referencing this geometry, accompanied by their rigid transformation matrices. This trick can significantly reduce the STEP file for a TPMS model, especially for those having many cells.

\section{Results and discussions}
\label{sec:results}
The proposed method has been implemented using C++ on a computer with an Intel Core i9-12900K CPU and 128GB RAM. 
Based on this implementation, nine case studies are to be presented to demonstrate the effectiveness of the proposed method. Case studies 1-3 (Fig.~\ref{result_tpms_structure}a-c) considered three simple uniform arrangements of typical TPMS structures, including Gyroid, Diamond, and Schwarz\_P. Case studies 4-6 (Fig.~\ref{result_tpms_structure}g-i) analyzed a more complex situation where TPMS structures were scaled to have distorted shapes, which increased the difficulty of surface approximation. Case studies 7-9 (Fig.~\ref{result_tpms_structure_intersection}) added even more complexity with trimmed TPMS structures. Error analysis (Figs.~\ref{result_error_analysis_1}-\ref{result_error_analysis_4} \rev{}{and Table~\ref{tab:error_tolerance_analysis}}) of the translated models in those case studies are given. 
\rev{Efficiency analysis (Figs.~\ref{result_time_tolerance} and~\ref{case_error_cpia})}{Continuity analysis (Table~\ref{tab:continuity_control}) and efficiency analysis (Figs.~\ref{result_time_tolerance} and~\ref{case_error_cpia})} are also provided. To test the validity of .step files after translation, Siemens NX (version 2306) has been used to open those files and then carry out model edits (e.g., boolean operations), and all case studies passed the test.

\begin{figure}[t]
  \centering
  \includegraphics[width=0.43\textwidth]{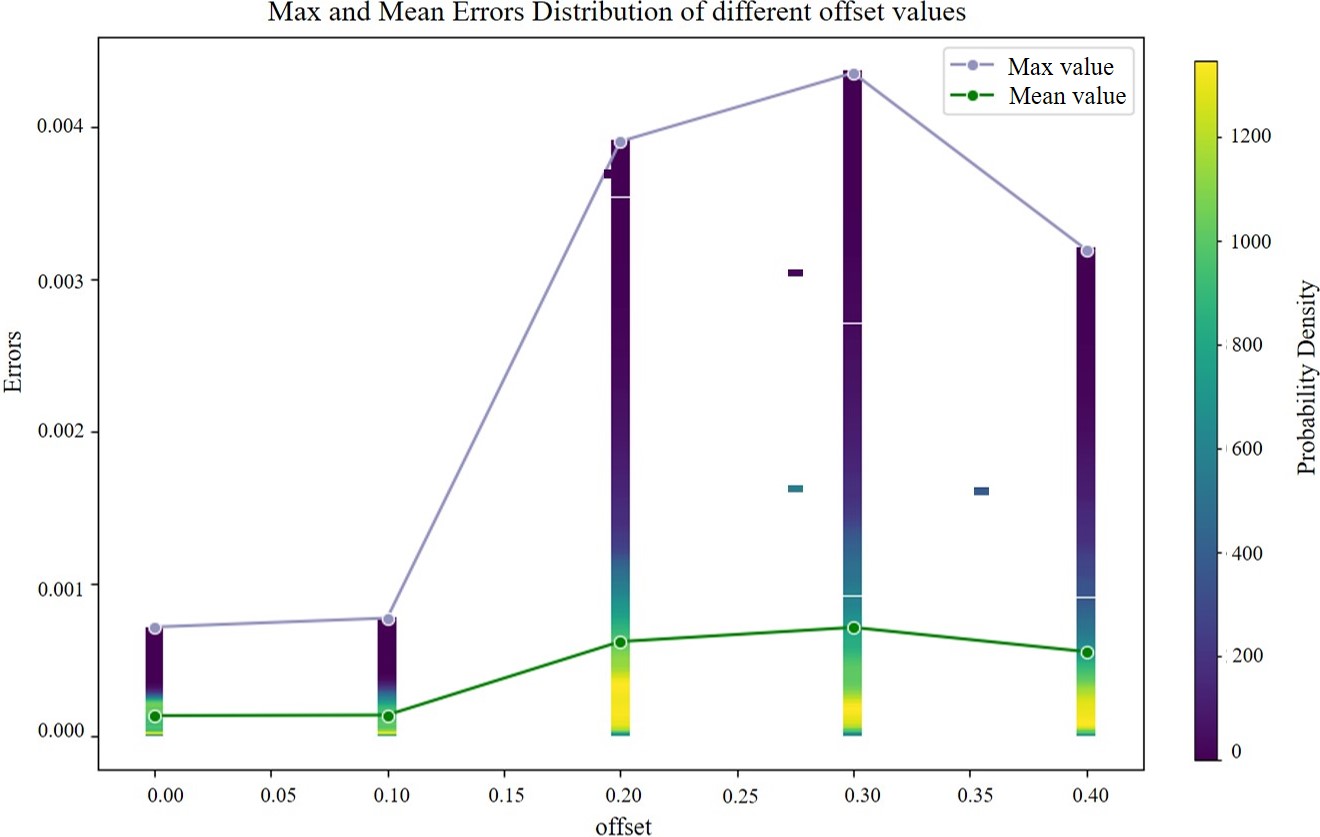}
  \caption{The statistical results of the error range distribution of Gyroid for error analysis.}
  \label{result_error_analysis_2}
\end{figure}
\begin{figure}[ht]
  \centering
  \includegraphics[width=0.43\textwidth]{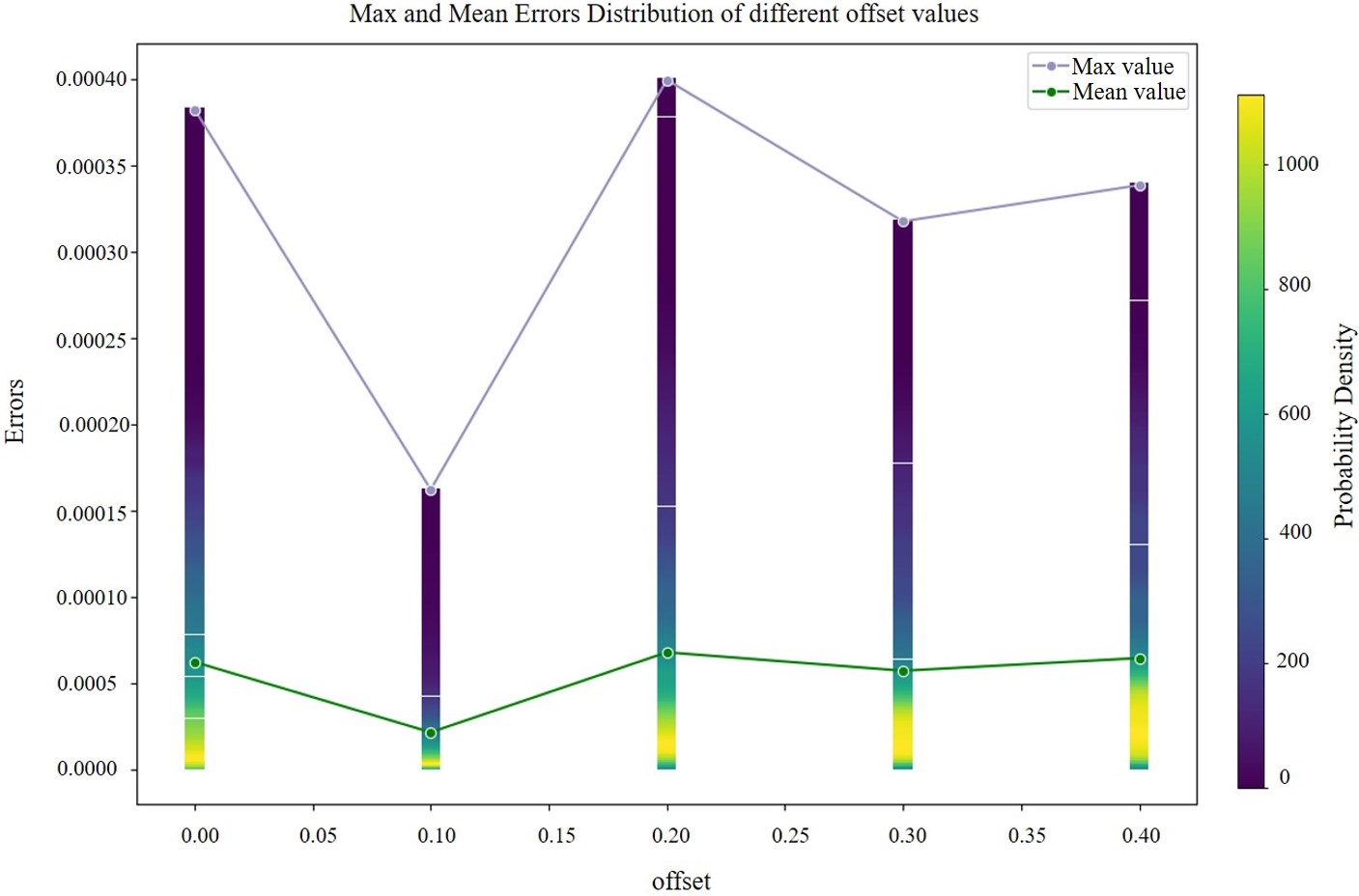}
  \caption{The statistical results of the error range distribution of Diamond for error analysis.}
  \label{result_error_analysis_3}
\end{figure}
\begin{figure}[t]
  \centering
  \includegraphics[width=0.45\textwidth]{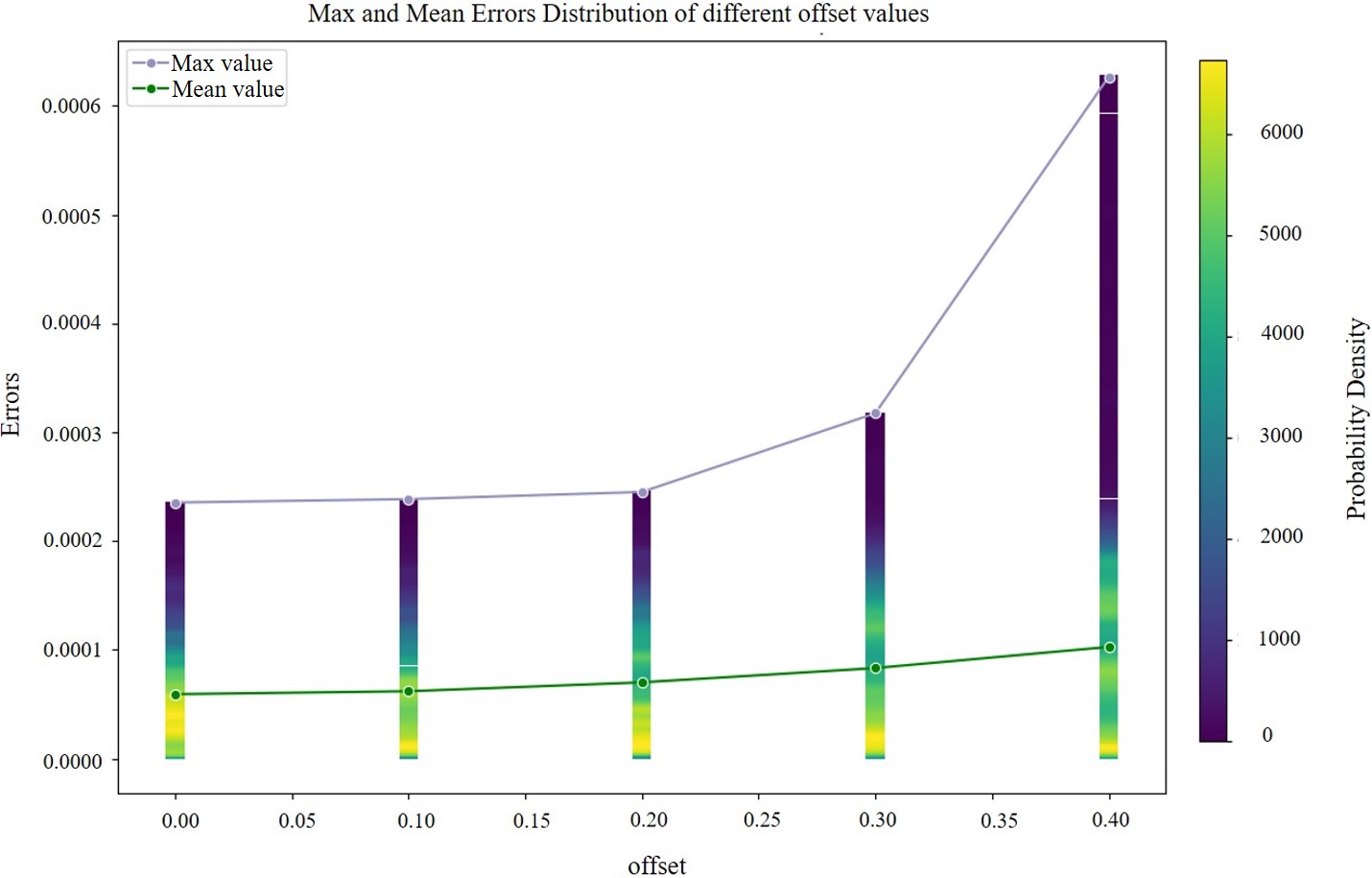}
  \caption{The statistical results of the error range distribution of Schwarz\_P for error analysis.}
  \label{result_error_analysis_4}
\end{figure}
\begin{figure}[t]
  \centering
  \includegraphics[width=0.45\textwidth]{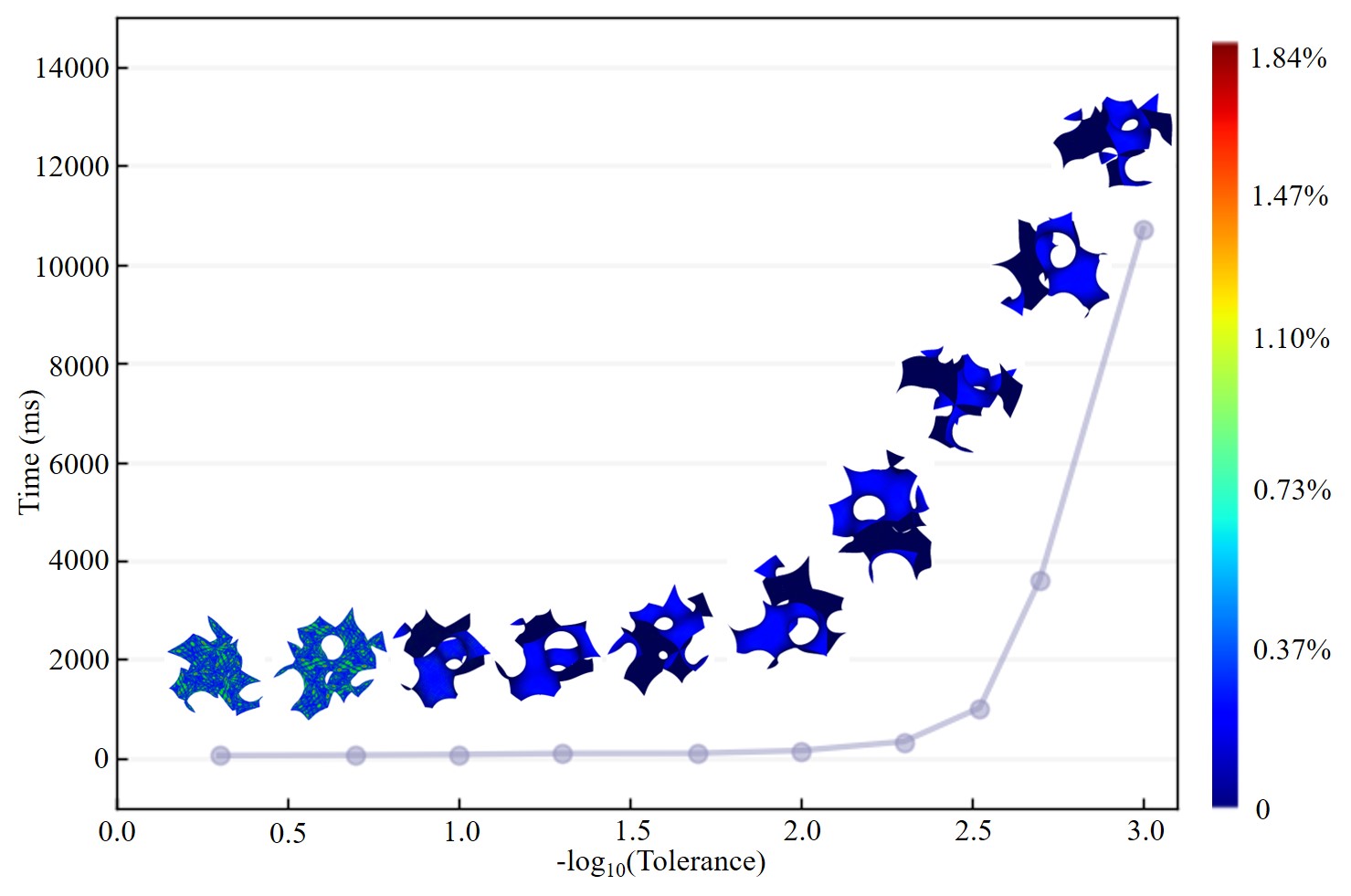}
  \caption{Time used in constructing NURBS surfaces under different tolerances for efficiency analysis.}
  \label{result_time_tolerance}
\end{figure}
\begin{figure}[t]
  \centering
  \includegraphics[width=0.45\textwidth]{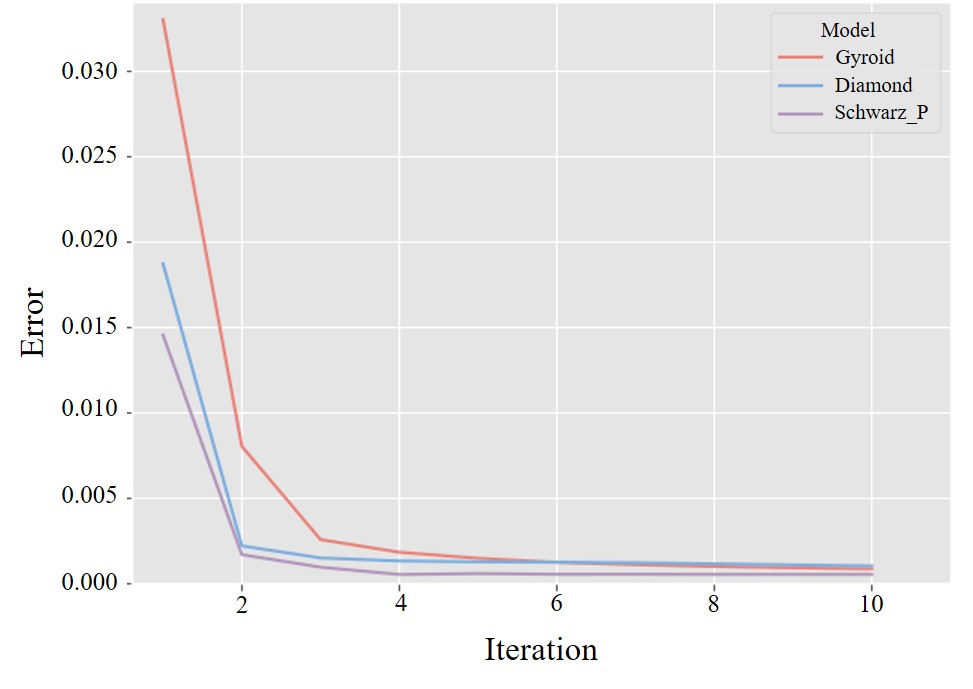}
  \caption{Error reduction in CPIA algorithm for efficiency analysis.}
  \label{case_error_cpia}
\end{figure}

\subsection{Examples}
\label{sec:examples}
Case studies 1-3 (Fig.~\ref{result_tpms_structure}~a-c) considered a TPMS unit model and an assembled model of Gyroid, Diamond, and Schwarz\_P TPMS respectively. Case studies 4-6 (Fig.~\ref{result_tpms_structure}~g-i) demonstrate the scaled TPMS solid models with distorted shapes. Case study 7-9 (Fig.~\ref{result_tpms_structure_intersection}) demonstrates the trimmed TPMS structures. \rev{}{Figs.~\ref{result_tpms_structure} and~\ref{result_tpms_structure_intersection} show the resulting surfaces/solids, and Table~\ref{tab:number_of_control_points} summarizes the numbers of control points of those approximate surfaces.}

Error analysis (Figs.~\ref{result_error_analysis_1}-\ref{result_error_analysis_4} \rev{}{and Table~\ref{tab:error_tolerance_analysis}}) considered the error distribution in translated models in the case studies. \rev{to prove the effectiveness of the error control process.
It also demonstrates the statistical results of the error range distribution.}{It also gives a comparison between the maximum error of the solid model and the given tolerance. From the statistics, the maximum errors are all below the given tolerances, which validates the error-controlling feature of our method.}
% It also gives a comparison between the maximum error of the solid model and the given tolerance.

\rev{}{\begin{table}
    \centering
    \footnotesize
    \setlength\extrarowheight{2pt}
    \setlength{\abovecaptionskip}{0cm}
    \caption{The offset, tolerance, and corresponding maximum errors between the TPMS solid model and the original TPMS model for error analysis.}
    \setlength{\tabcolsep}{4mm} {
    \begin{tabular}{l l l l l}
    \hline
         \multirow{2}{*}{Offset} & \multirow{2}{*}{Tolerance} & \multicolumn{3}{c}{Maximum Errors} \\
         \cline{3-5}
                            &    & Gyroid & Diamond & Schwarz\_P \\  \hline
         0 & 0.005 &  0.00024 & 0.0018 &  0.00022 \\
         0.1 & 0.01 & 0.00077  & 0.0016 & 0.00041 \\
         0.2 & 0.01 & 0.0039 & 0.0040 & 0.00025 \\
         0.3 & 0.01 & 0.0044  & 0.0032 & 0.00032 \\
         0.4 & 0.01 & 0.0032  & 0.0034 & 0.00063 \\
         \hline
    \end{tabular} }
    \label{tab:error_tolerance_analysis}
\end{table}}

\rev{}{\begin{table}
    \centering
    \setlength\extrarowheight{3pt}
    \setlength{\abovecaptionskip}{0cm}
    \footnotesize
    \caption{The relative maximum and average differences of the first-order and second-order derivatives between two connected edges of two adjacent approximation surfaces for continuity analysis.}
    \setlength{\tabcolsep}{1mm} {
    \begin{tabular}{l l l l l l l l}
    \hline
         \multirow{2}{*}{Model Type} & \multirow{2}{*}{Offset} & \multicolumn{2}{c}{Maximum Differences} & \multicolumn{2}{c}{Average Differences} \\
         \cline{3-6}
                    & &  First Order & Second Order & First Order & Second Order \\
        \hline
        Gyroid & 0 & 0.08\% & 0.25\% & 0.04\% & 0.18\% \\
        Gyroid & 0.3 & 0.12\% & 0.60\% & 0.06\% & 0.40\% \\
        Diamond & 0 & 0.06\%  & 0.15\% & 0.09\% & 0.13\% \\
        Diamond & 0.3 & 0.07\% & 0.15\% & 0.09\% & 0.13\% \\
        Schwarz\_P & 0 & 0.08\% & 0.82\% & 0.04\% & 0.42\% \\
        Schwarz\_P & 0.3 & 0.09\% & 0.80\% & 0.04\% & 0.44\% \\
    \hline
    \end{tabular} }
    \label{tab:continuity_control}
\end{table}}

\rev{}{\begin{table*}
    \centering
    \setlength\extrarowheight{3pt}
    \setlength{\abovecaptionskip}{0cm}
    \footnotesize
    \caption{The numbers of control points in approximation surfaces of case studies (1)-(9).}
    \setlength{\tabcolsep}{3mm} {
    \begin{tabular}{l l c c}
    \Xhline{1.2pt}
        Model Type & Offset & Number of Sample Points & Number of Control Points \\
        \Xhline{1.2pt}
            \multirow{5}{*}{Gyroid} & 0 & 24025 & 24025 \\
             & 0.1 & 10609 & 10609 \\
             & 0.2 & 10609  & 10609 \\
             & 0.3 & 10609  & 10609 \\
             & 0.4 & 10609  & 10609 \\
             \hline
             \multirow{5}{*}{Diamond} & 0 & 2601 & 2601 \\
             & 0.1 & 2601 & 2601 \\
             & 0.2 & 2601  & 2601 \\
             & 0.3 & 2601  & 2601 \\
             & 0.4 & 2601  & 2601 \\
             \hline
             \multirow{5}{*}{Schwarz\_P} & 0 & 2601 & 2601 \\
             & 0.1 & 2601 & 2601 \\
             & 0.2 & 2601  & 2601 \\
             & 0.3 & 2601  & 2601 \\
             & 0.4 & 2601  & 2601 \\
            % Gyroid & 0.1 & 10609 & Diamond & 0.1 & 2601 & Schwarz\_P & 0.1 & 2601 \\
            % Gyroid & 0.2 & 10609 & Diamond & 0.2 & 2601 & Schwarz\_P & 0.2 & 2601 \\
            % Gyroid & 0.3 & 10609 & Diamond & 0.3 & 2601 & Schwarz\_P & 0.3 & 2601 \\
            % Gyroid & 0.4 & 10609 & Diamond & 0.4 & 2601 & Schwarz\_P & 0.4 & 2601 \\
         \Xhline{1.2pt}
    \end{tabular} }
    \label{tab:number_of_control_points}
\end{table*}}

\rev{}{Continuity analysis (Table~\ref{tab:continuity_control}) analyzed the continuity between adjacent approximation surfaces in translated models in the case studies. The relative maximum and average differences of the first-order and second-order derivatives between two connected edges of two adjacent surfaces are calculated to validate the $C^2$ continuity-preserving feature.}

Efficiency analysis (Figs.~\ref{result_time_tolerance} and ~\ref{case_error_cpia}) demonstrated the construction time of error-controlled NURBS surfaces with different tolerances to prove the efficiency of the proposed method and showed the error reduction process of CPIA to verify the convergence of the CPIA.
 
% To compare the CPIA algorithm for NURBS approximation with the method in OpenCascade, we conducted tests using different scales of sample points for approximation in case study 8 (refer to Table \ref{tab:approximation_time}).
% \begin{table*}
%     \centering
%     \begin{tabular}{c|c|c|c}
%     \hline
%          Number of sample points & Time with OCC(ms) & Time with CPIA(ms) & Speed up \\
%          \hline
%          169 & 3.46 & 0.51 &  $6.88\times$ \\
%          225 & 3.75 & 0.63 & $5.94\times $\\
%          1,849 & 17.43 & 4.54 & $3.84\times$ \\
%          24,025 & 119.28 & 55.04 & $2.16\times$ \\
%          1,485,961 & 7,642.85 & 3,407.27 & $ 2.24\times$ \\
%          \hline
%     \end{tabular}
%     \caption{Time used in NURBS approximation for CPIA(CPIA) and the method in OpenCasCade(OCC)}
%     \label{tab:approximation_time}
% \end{table*}

\subsection{Discussion and Limitations}
\label{sec:discussion}
In all of the above TPMS solid models, the NURBS surfaces are constructed with an error less than the specified tolerance, confirming the effectiveness of the error-driven TPMS sampling algorithm. The high-accuracy NURBS surface with an error upper bound of 1mm could be constructed within 10.7s.
From Fig.~\ref{result_error_analysis_1}, the error is mainly found in specific parts of the surface, which may be influenced by the sampling strategy. From Fig.~\ref{result_error_analysis_2}, Fig.~\ref{result_error_analysis_3}, and Fig.~\ref{result_error_analysis_4}, the proportion of these regions with large errors is quite low.
\rev{In Fig.~\ref{result_tpms_structure}, the variation of the second-order partial derivative shows the $ C^2$ continuity maintenance of our method.}{}
% In Table ~\ref{tab:error_tolerance_analysis}, the maximum errors are all below the corresponding tolerance. Here the maximum errors for Schwarz\_P are significantly lower than the tolerance, which means our sampling approach may have room for improvement for some special cases. The variation of the second-order partial derivative in Fig.~\ref{result_tpms_structure} and the continuity analysis results in Table ~\ref{tab:continuity_control} show the $ C^2$ continuity maintenance of our method.
In Fig.~\ref{case_error_cpia}, the reduced error during CPIA shows the convergence of CPIA from an experimental perspective and further confirms the effectiveness of the convergence proof in Sec.~\ref{sec:proof}.

\rev{}{From the error stats in Table~\ref{tab:error_tolerance_analysis}, we can find that the error bound $2\epsilon$ derived in Section 3.2 is somewhat conservative. For example, the empirical maximum error is 0.0044 in the case with $offset=0.3$ for Gyroid, but the theoretical tolerance is 0.01. For Schwarz\_P with $offset=0.3$, the maximum error is only $3.2\%$ of the tolerance, which implies that a tighter bound may be derived in future work.}

\rev{}{In Fig.~\ref{result_tpms_structure}d-f, the variation of the second-order partial derivative demonstrates the continuous distribution of the second-order partial derivatives, and in Table~\ref{tab:continuity_control}, the relative differences of the first-order and second-order derivatives between the connected edges of two adjacent surfaces are no more than 0.12\% and 0.82\%, respectively. They both show the $ C^2$ continuity feature of our method.}

In this paper, although only the three most commonly used TPMS types are considered, this method could also be applied to part of other types of TPMS structures like I-WP or Schwarz\_S. They have different Weierstrass functions from the three types of TPMS structures utilized in this paper~\cite{Fogden1992}. However, the Weierstrass functions cannot cover all TPMS \rev{model}{models} (e.g., the $\mathbf{C}(\mathbf{P})$ surface)~\cite{2000_Gandy_TPMS_explicit_representation_Weierstrass_Gyroid}. This method cannot be applied to TPMS structures with unknown Weierstrass functions. This is a serious limitation of the proposed method.

\section{Conclusion}
\label{sec:conclusion}
A new method has been presented in this paper to translate TPMS to STEP. The main features of this method include error control and continuity preservation in the translation. These main features are essentially achieved through two new algorithms: an approximation error-driven TPMS sampling algorithm and a CPIA algorithm. When working together, they can give an error-controlled, $C^2$ continuity-preserving, and fast translation of TPMS models to STEP files. Both the theoretical proof of the proposed method's convergence and the empirical validation of the proposed method's effectiveness have been demonstrated by using a series of examples and comparisons.

Although the proposed method is seen to be quite effective in the case studies and analyses conducted, there are still a few limitations. Certain types of TPMS cannot be represented by Weierstrass functions (e.g., the $\mathbf{C}(\mathbf{P})$ surface)~\cite{2000_Gandy_TPMS_explicit_representation_Weierstrass_Gyroid}, thus falling outside the applicability scope of the proposed method. Another limitation is that the CPIA algorithm requires the same number of control points as the sample points, leading to a large number of control points when a high TPMS translation accuracy is required. Extending the present work to overcome this limitation can be very practically beneficial, and are among the CAD research studies to be carried out in the future. The LSPIA algorithm~\cite{2018_lin_survey_pia} may help here, but further development is needed to make its approximation error controllable. \rev{}{Surface approximation with parameter correction, e.g., those presented in~\cite{2012_kineri_surface_fitting_iterative,2024_bo_iterative_approximation_error_control_parameter_correction}, can be helpful in attaining lower approximation errors for PIA and is among our future research studies.}

\section*{Acknowledgments}
This work has been funded by NSF of China (No. 62102355), the ``Pioneer" and ``Leading Goose" R\&D Program of Zhejiang Province (No. 2024C01103), NSF of Zhejiang Province (No. LQ22F020012), and the Fundamental Research Funds for the Zhejiang Provincial Universities (No. 2023QZJH32).

% \appendix
% \section{The second derivative of Eq.~\ref{offsetsurface}}

% \label{append-extend}
% %The descriptions in this appendix are to be made brief since they are secondary results of this work.

% \section*{References}

\bibliographystyle{elsarticle-num}
\bibliography{mybibfile}

% begin biography
% \newpage

% \section*{}
% \noindent\textbf{Biography}\\
% \begin{figure}[htb]
% \includegraphics[width=0.12\textwidth]{protrait.jpg}
% \end{figure}

% \noindent\textbf{Qiang Zou} is currently an assistant professor of Computer Science at Zhejiang University, China. He got his Ph.D. from the University of British Columbia, Canada, in 2019. His research interests lie in design modeling and manufacturing simulation.

\end{document}